\documentclass[12pt, onecolumn, draftcls]{IEEEtran}

\usepackage{epsfig, psfrag,latexsym,amsmath,epsf,amssymb,bm}
\usepackage{cite}

\newtheorem{theorem}{Theorem}
\newtheorem{corollary}{Corollary}
\newtheorem{lemma}{Lemma}

\newtheorem{Def}{Definition}
\newtheorem{rem}{Remark}

\begin{document}
\title{Interference Channels with Correlated \linebreak Receiver Side Information \thanks{This research was supported in part by the U.S. National Science Foundation under Grants  ANI-03-38807, CCF-07-28208, and CNS-06-25637, the DARPA ITMANET program under Grant 1105741-1-TFIND, and the U.S. Army Research Office under MURI award W911NF-05-1-0246.}}
\author{Nan Liu$^1$, D. G\"{u}nd\"{u}z$^{1,2}$, A. Goldsmith$^{1}$ and H. V. Poor$^{2}$\\
$^{1}$ Dept. of Electrical Engineering, Stanford Univ., Stanford, CA 94305, USA\\
\and $^{2}$ Dept. of Electrical Engineering, Princeton Univ., Princeton, NJ 08544, USA\\
\and \texttt{{\small \{nanliu@stanford.edu, dgunduz@princeton.edu, andrea@stanford.edu, poor@princeton.edu\}.}}}
%
%
%
\maketitle

\begin{abstract}
The problem of joint source-channel coding in transmitting independent sources over interference channels with correlated receiver side information is studied. When each receiver has side information correlated with its own desired source, it is shown that source-channel code separation is optimal. When each receiver has side information correlated with the interfering source, sufficient conditions for reliable transmission are provided based on a joint source-channel coding scheme using the superposition encoding and partial decoding idea of Han and Kobayashi. When the receiver side information is a deterministic function of the interfering source, source-channel code separation is again shown to be optimal. As a special case, for a class of Z-interference channels, when the side information of the receiver facing interference is a deterministic function of the interfering source, necessary and sufficient conditions for reliable transmission are provided in the form of single letter expressions. As a byproduct of these joint source-channel coding results, the capacity region of a class of Z-channels with degraded message sets is also provided.
\end{abstract}

\vspace*{0.4cm}
\noindent {\em Index terms:}
Interference channel, joint source-channel coding,
receiver side information, source-channel separation theorem

\newpage

\section{Introduction}\label{s:introduction}

The wireless medium is shared by multiple communication systems operating simultaneously, which leads to interference among users transmitting over the same frequency band. In the simple scenario of two transmitter-receiver pairs, the interference channel \cite{Shannon:1961} models two simultaneous transmissions interfering with each other. In the classical interference channel model, the sources intended for each receiver are independent of each other, and the receivers decode based only on their own received signals. On the other hand, in applications such as sensor networks, it is reasonable to assume that the receivers have access to their own correlated observations about the underlying source sequences as well. These correlated observations at the receivers can be exploited to improve the system performance.

Even in the absence of side information, a finite letter expression for the capacity region of an interference channel in the general case is unknown. We have the capacity region in the case of interference channels with statistically equivalent
outputs \cite{Ahlswede:1971,Sato:1977,Carleial:1978}, discrete additive degraded
interference channels \cite{Benzel:1979}, a class of
deterministic interference channels \cite{ElGamal:1982}, strong interference channels
\cite{Carleial:1975, Sato:1978MayIC, Han:1981, Sato:1981,Costa:1987}, a class of degraded interference channels \cite{Liu_Ulukus:2006Allerton}, and more recently for a class of Z-interference channels \cite{Liu_Goldsmith:2008ISIT}. The best known achievable rate region is due to Han and Kobayashi \cite{Han:1981}, a simplification of which is given in \cite{Chong:2006}.

In a point-to-point scenario, the availability of correlated side information at the receiver is considered in \cite{Shamai:ETT:95}. It is shown that the source-channel separation theorem applies in this simple setting and, moreover, that Slepian-Wolf source coding followed by channel coding is optimal. With the availability of side information at the receiver, we can transmit the source reliably over a channel with smaller capacity than the one required when there is no receiver side information. However, it is known that the source-channel separation theorem does not generalize to multi-user channels \cite{Shannon:1961}, \cite{Cover:1980}, and necessary and sufficient conditions for reliable transmission in the case of correlated sources and correlated receiver side information are not known in general. In \cite{Tuncel:IT:06}, necessary and sufficient conditions are characterized for broadcasting a common source to multiple receivers with different correlated side information. An alternative achievability scheme for the setup of \cite{Tuncel:IT:06} is given in \cite{Gunduz:ITW:07}. In \cite{Kramer:2007}, the results of \cite{Tuncel:IT:06} are extended to broadcast channels with degraded message sets in which the receivers have access to parts of the underlying messages. Availability of messages or message parts at the receivers of broadcast channels from the channel coding perspective is studied also in \cite{Wu:2007, Xie:2007, Xue:2007}. In \cite{Kang:ISIT:08}, broadcasting a pair of correlated sources with correlated receiver side information is studied.

The interference channel with correlated sources is considered in \cite{Salehi:1993}, and a sufficient condition for reliable transmission is given. In \cite{Gunduz:IT:08}, an interference channel with independent sources, in which each receiver has access to side information correlated with the interfering transmitter's source, is considered. Necessary and sufficient conditions for this setup are characterized under the strong source-channel interference conditions, which generalize the usual strong interference conditions by considering correlated side information as well. The result of \cite{Gunduz:IT:08} shows that interference cancellation is optimal even when the underlying channel interference is not strong, as long as the overall source-channel interference is.

In this paper, we extend the scenario studied in \cite{Gunduz:IT:08} to more general interference channels. We first consider the case in which each receiver has side information correlated with the source sequence it wants to decode. We prove the optimality of source-channel code separation in this situation; that is, the optimal performance can be achieved by first compressing each of the sources using Slepian-Wolf coding with respect to the correlated receiver side information, and then transmitting the compressed bits over the channel using an optimal interference channel code.

Next, we consider the scenario in which each receiver has side information correlated with the interfering transmitter's source. As an example of such a model and to illustrate the benefits of side information about the interfering source, consider the extreme case in which each receiver has access to the message of the interfering transmitter. Note that this setup is equivalent to the restricted two-way channel model of Shannon, whose capacity is characterized in \cite{Shannon:1961}. In this case, each receiver can excise the interference from the undesired transmitter, since its message is exactly known at the receiver. Here, we consider the more general case of arbitrary correlation between the receiver side information and the interfering source, and propose a joint source-channel coding scheme similar to that of Han and Kobayashi \cite{Han:1981} taking the side information into account. Later, we consider the case in which the side information is a deterministic function of the interfering source, and show that source-channel code separation is again optimal. Finally, we consider a special class of interference channels called Z-interference channels, in which only one receiver faces interference. Further focusing on a special class of Z-interference channels satisfying certain conditions (which will be stated later), and the case in which the side information is a deterministic function of the interfering source, we are able to characterize necessary and sufficient conditions for reliable transmission in the form of single letter expressions. This setting also constitutes an example for which the general sufficiency conditions we provide are also necessary, proving their tightness for certain special cases.

The rest of the paper is organized as follows. In Section \ref{system_model} we present the system model.  In Section \ref{side_desired} we prove the optimality of source-channel code separation when the side information is correlated with the desired source. The case in which the side information is correlated with the interfering source is considered in Section \ref{side_interference}. In Section \ref{sufficient_conditions}, we provide sufficient conditions for reliable transmission, while in Section \ref{sepdet}, we prove the optimality of source-channel code separation when the side information is a deterministic function of the interfering source. In Section \ref{special} we show that, for a special source and channel model, the sufficient conditions for reliable transmission proposed in Section \ref{sufficient_conditions} are also necessary, and hence we give a single letter characterization of the necessary and sufficient conditions for this model. In Section \ref{s:degmsg_Z} we characterize the capacity region of a class of Z-channels with degraded message sets. This is followed by conclusions in Section \ref{conclusions}.

\section{System Model} \label{system_model}

An interference channel is composed of two transmitter-receiver pairs. The underlying discrete memoryless channel is characterized by the transition probability $p(y_1, y_2|x_1,x_2)$ from finite input alphabet $\mathcal{X}_1 \times \mathcal{X}_2$ to finite output alphabet $\mathcal{Y}_1 \times \mathcal{Y}_2$. Transmitter $k$ has access to the source sequence $\{U_{k,i}\}_{i=1}^\infty$, $k=1,2$. Consider side information sequences $\{V_{k,i}\}_{i=1}^\infty$, where the source and the side information sequences are independent and identically distributed (i.i.d.) and are drawn according to joint distribution $p(u_1, v_1) p(u_2, v_2)$ over a finite alphabet $\mathcal{U}_1 \times \mathcal{V}_1 \times \mathcal{U}_2 \times \mathcal{V}_2$; that is, the two source-side information pairs are independent of each other.

For $k=1,2$, Transmitter $k$ observes $U_k^n$ and wishes to transmit it noiselessly to Receiver $k$ over $n$ uses of the channel\footnote{Here we use the notation $U_k^n=(U_{k,1}, \ldots, U_{k,n})$, and similar notation for other length-$n$ sequences.}. The encoding function at Transmitter $k$ is
\[f_k^n: \mathcal{U}_k^n \rightarrow \mathcal{X}_k^n. \]

We assume that the side information $V_{\pi(k)}^n$ is available at receiver $k$, where $\pi(\cdot)$ is a permutation of $\{1,2\}$. Depending on the scenario, we will specify whether the side information is correlated with the desired source or with the interfering source.

The decoding function at receiver $k$ reconstructs its estimate $\hat{U}_k$ from its channel output and side information vector using the decoding function
\[g_k^n : \mathcal{Y}_k^n \times \mathcal{V}_{\pi(k)}^n \rightarrow \mathcal{U}_k^n. \]

The probability of error for this system is defined as
\begin{eqnarray}\label{prob_error}
P_e^n &=& \text{Pr}\{(U_1^n, U_2^n) \neq (\hat{U}_1^n, \hat{U}_2^n)\}, \nonumber \\
   &=& \sum_{(u_1^n, u_2^n) \in \mathcal{U}_1^n\times \mathcal{U}_2^n} p(u_1^n, u_2^n) P\left\{ (\hat{U}_1^n, \hat{U}_2^n) \neq (u_1^n, u_2^n) \big| (U_1^n,U_2^n) = (u_1^n,u_2^n)\right\}. \nonumber
\end{eqnarray}

\begin{Def}
We say that a source pair $(U_1, U_2)$ can be reliably transmitted over a given interference channel if there exist a sequence of encoders and decoders $(f_1^n, f_2^n, g_1^n, g_2^n)$ such that $P_e^n \rightarrow 0$ as $n \rightarrow \infty$.
\end{Def}

In the following sections, we consider two cases in particular. In the first case, each receiver has side information correlated with its desired source, i.e., $\pi(k) =k$, $k=1,2$. In the second case, each receiver has side information correlated with the interfering source, i.e., $\pi(1) = 2$ and $\pi(2)=1$. In both cases, we want to exploit the availability of correlated side information at the receivers. In the first case, each transmitter needs to transmit less information to its intended receiver due to the availability of correlated side information. In the latter case, the side information is used to mitigate the effects of interference.

For notational convenience, we drop the subscripts on probability
distributions unless the arguments of the distributions are not lower case
versions of the corresponding random variables.

\begin{figure*}
\centering
\psfrag{X1}{$X_1$}\psfrag{X2}{$X_2$}
\psfrag{Y1}{$Y_1$}\psfrag{Y2}{$Y_2$}
\psfrag{U1}{$U_1$}\psfrag{U2}{$U_2$}
\psfrag{hU1}{$\hat{U}_1$}\psfrag{hU2}{$\hat{U}_2$}
\psfrag{V1}{$V_1$}\psfrag{V2}{$V_2$}
\psfrag{Tx1}{$\mathrm{Tx}_1$} \psfrag{Tx2}{$\mathrm{Tx}_2$}
\psfrag{Rx1}{$\mathrm{Rx}_1$} \psfrag{Rx2}{$\mathrm{Rx}_2$}
\psfrag{p1}{$p(u_1, v_1)$}\psfrag{p2}{$p(u_2, v_2)$}
\psfrag{pxy}{$p(y_1, y_2|x_1, x_2)$}
\includegraphics[width=5in]{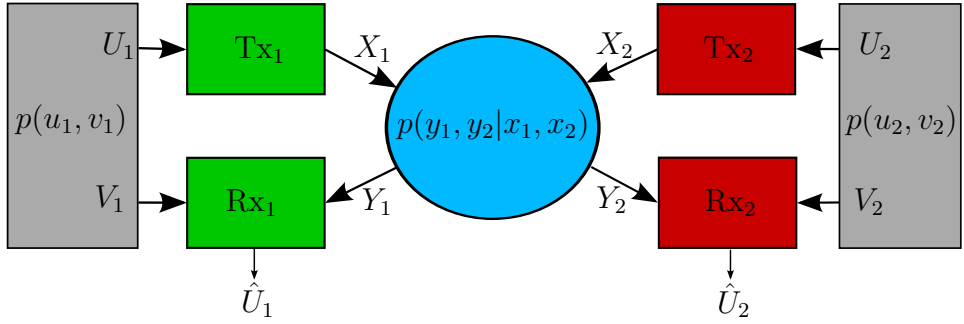}
\caption{Interference channel model in which the receivers have access to side information correlated with the source they want to receive.} \label{f:system_model}
\end{figure*}

\section{Side information correlated with the desired source} \label{side_desired}

In this section, we consider an interference channel in which each receiver has side information correlated with the source it wants to decode, i.e., receiver $k$ has access to side information $V_k$ (see Fig. \ref{f:system_model}). For this special case, we prove that the source-channel separation theorem applies; that is, it is optimal for the transmitters first to apply Slepian-Wolf source coding to compress their sources conditioned on the side information at the corresponding receiver, and then to transmit the compressed bits over the channel using an optimal interference channel code. Note that, in the general case, we do not have a single-letter characterization of the capacity region of the interference channel, yet we can still prove the optimality of source-channel code separation. In the proof, we use the $n$-letter expression for the capacity region, which was also used in \cite{Gunduz:DCC:07} to prove the optimality of source-channel code separation for a multiple access channel with receiver side information and feedback. The main result of this section is the following theorem.

\begin{theorem} \label{ownSideinfo}
Sources $U_1$ and $U_2$ can be transmitted reliably to their respective receivers over the discrete memoryless interference channel $p(y_1,y_2|x_1,x_2)$ with side information $V_k$ at receiver $k$, $k=1,2$, if
\begin{align}
(H(U_1|V_1), H(U_2|V_2)) \in int( \mathcal{C}) \label{ownSideagain}
\end{align}
where $int(\cdot)$ denotes the \emph{interior}, and $\mathcal{C}$ denotes the capacity region of the underlying interference channel.

Conversely, if $(H(U_1|V_1), H(U_2|V_2)) \notin \mathcal{C}$, then sources $U_1$ and $U_2$ cannot be transmitted reliably.
\end{theorem}
\begin{proof}
A proof of Theorem \ref{ownSideinfo} is given in Appendix \ref{proof_ownSideinfo}.
\end{proof}

\section{Side information correlated with the interfering source} \label{side_interference}

In this section we consider the case in which Receiver 1 has access to $V_2$ while Receiver 2 has access to $V_1$, i.e., each receiver has side information about the interfering transmitter's source (see Fig. \ref{f:system_model2}). We investigate how the side information about the interference helps in decoding the desired information.

\begin{figure*}
\centering
\psfrag{X1}{$X_1$}\psfrag{X2}{$X_2$}
\psfrag{Y1}{$Y_1$}\psfrag{Y2}{$Y_2$}
\psfrag{U1}{$U_1$}\psfrag{U2}{$U_2$}
\psfrag{hU1}{$\hat{U}_1$}\psfrag{hU2}{$\hat{U}_2$}
\psfrag{V1}{$V_1$}\psfrag{V2}{$V_2$}
\psfrag{Tx1}{$\mathrm{Tx}_1$} \psfrag{Tx2}{$\mathrm{Tx}_2$}
\psfrag{Rx1}{$\mathrm{Rx}_1$} \psfrag{Rx2}{$\mathrm{Rx}_2$}
\psfrag{p1}{$p(u_1, v_1)$}\psfrag{p2}{$p(u_2, v_2)$}
\psfrag{pxy}{$p(y_1, y_2|x_1, x_2)$}
\includegraphics[width=5in]{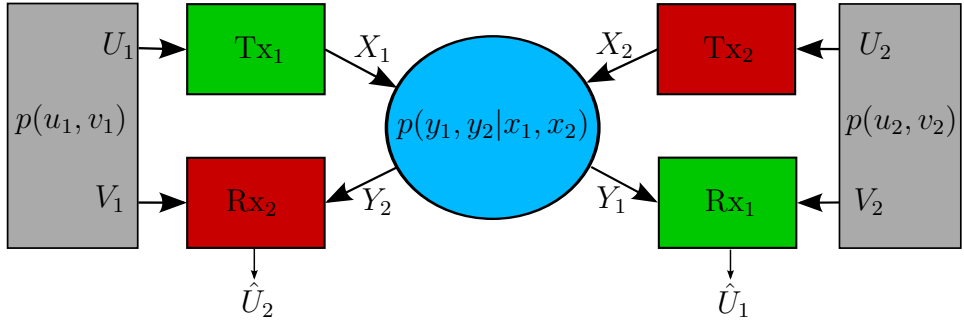}
\caption{Interference channel model in which the receivers have access to side information correlated with the source of the interfering transmitter.} \label{f:system_model2}
\end{figure*}

\subsection{Sufficient Conditions for Reliable Transmission} \label{sufficient_conditions}

We first provide sufficient conditions for reliable transmission of the sources. In the spirit of the Han-Kobayashi scheme for the classical interference channel, we propose a joint source-channel coding scheme that requires the receivers to decode part of the interference with the help of their side information. In the Han-Kobayashi scheme, each transmitter splits its message into two pieces to allow the non-intended receiver to decode part of the interference. In our scheme, each transmitter enables a quantized version of its source to be decoded by both receivers, where the unintended receiver uses its correlated side information as well as the channel output to decode the interference corresponding to this quantized part. Sufficient conditions for reliable transmission in this setup are given in the following theorem.

\begin{theorem} \label{wish}
Sources $U_1$ and $U_2$ can be transmitted reliably over the interference channel $p(y_1,y_2|x_1,x_2)$ with side information $V_1$ at Receiver 2 and $V_2$ at Receiver 1 if there exist random variables $W_1$ and $W_2$ such that
\begin{align}
H(U_1)  <& I(X_1;V_2,Y_1|W_2, Q), \\
H(U_2) < & I(X_2;V_1,Y_2|W_1, Q), \\
H(U_1) < & I(W_2, X_1;V_2,Y_1|Q)-I(U_2;W_2|Q), \\
H(U_2) < & I(W_1, X_2;V_1,Y_2| Q)-I(U_1;W_1|Q), \\
H(U_1)+H(U_2) < & I(X_1;V_2,Y_1|W_1,W_2, Q)+  I(W_1, X_2;V_1,Y_2|Q),  \label{r1}\\
H(U_1)+H(U_2) < & I(X_2;V_1,Y_2|W_1,W_2,Q)+I(W_2, X_1;V_2,Y_1|Q), \\
H(U_1)+H(U_2) < & I(W_1, X_2;V_1, Y_2|W_2,Q) +I(W_2, X_1;V_2, Y_1|W_1,Q), \label{r2}\\
H(U_1)+H(U_2) < & I(W_2,X_1;V_2,Y_1|Q) +I(W_1,X_2;V_1,Y_2|W_2,Q)-I(U_1;W_1|Q), \\
H(U_1)+H(U_2)  < & I(W_1,X_2;V_1,Y_2|Q)+I(W_2,X_1;V_2,Y_1|W_1,Q)-I(U_2;W_2|Q),\\
2H(U_1)+H(U_2) < & I(W_2, X_1;V_2,Y_1|Q) + I(X_1;V_2,Y_1|W_1,W_2,Q) + I(W_1, X_2;V_1, Y_2|W_2,Q), \\
H(U_1)+2H(U_2) < &  I(W_1, X_2;V_1,Y_2|Q)+I(X_2;V_1,Y_2|W_1,W_2,Q) +I(W_2, X_1;V_2, Y_1|W_1,Q),
\end{align}
for some $p(q)$, $p(w_1, x_1|u_1,q)$, and $p(w_2,x_2|u_2,q)$, where the entropies and mutual information terms are evaluated using the joint distribution
\begin{align}
p(q, u_1,v_1,u_2,v_2,w_1,w_2,x_1,x_2,y_1,y_2)= p(q) p(u_1,v_1)p(u_2,v_2)p(w_1,x_1|u_1,q) & p(w_2,x_2|u_2,q) \nonumber \\
& p(y_1, y_2|x_1,x_2).\label{dis}
\end{align}
\end{theorem}
\begin{proof}
A proof of Theorem \ref{wish} is given in Appendix \ref{proofsufficient}.
\end{proof}

We remark here that the achievability scheme in the proof of Theorem \ref{wish} uses joint source-channel coding and hence, similarly to \cite{Cover:1980} and \cite{Han:1987}, the expressions involve joint distribution of the source and channel variables, which potentially increases the achievable rate region by enlarging the set of possible joint distributions.
Below in Corollary \ref{c:separate}, we provide a sufficient condition for reliable transmission based on separate source and channel codes in the spirit of ``operational separation'' as in \cite{Tuncel:IT:06}, \cite{Gunduz:IT:08}, which can be obtained as a special case of Theorem \ref{wish}. Note that operational separation is different from the classical (``informational'') separation, in which each source is first assigned to an index and then these indices are transmitted using an optimal channel code for the underlying channel. Operational separation corresponds to separation of the source and the channel variables as in Corollary \ref{c:separate} without using the optimal source or the channel codes (see \cite{Gunduz:IT:08} for further details and examples).

\begin{corollary} \label{c:separate}
Sources $U_1$ and $U_2$ can be transmitted reliably over the interference channel $p(y_1,y_2|x_1,x_2)$ with side information $V_1$ at Receiver 2 and $V_2$ at Receiver 1 if there exist random variables $\overline{W}_1, \widetilde{W}_1$ and $\overline{W}_2, \widetilde{W}_2$ such that
\begin{align}
H(U_1)  <& I(X_1; Y_1|\widetilde{W}_2, Q),   \\
H(U_1) + I(\overline{W}_2; U_2 |V_2, Q) < & I(X_1, \widetilde{W}_2; Y_1|Q),  \\
H(U_2) < & I(X_2; Y_2|\widetilde{W}_1, Q),  \\
H(U_2) + I(\overline{W}_1; U_1 |V_1,Q) < & I(X_2, \widetilde{W}_1; Y_2| Q),   \\
H(U_1)+H(U_2) - I(\overline{W}_1;V_1|Q) <&  I(X_1; Y_1|\widetilde{W}_1,\widetilde{W}_2, Q) +  I(\widetilde{W}_1, X_2; Y_2|Q),   \\
H(U_1)+H(U_2) - I(\overline{W}_2;V_2|Q) < & I(X_2; Y_2|\widetilde{W}_1,\widetilde{W}_2, Q) +  I(\widetilde{W}_2, X_1; Y_1|Q),  \\
H(U_1)+H(U_2) - I(\overline{W}_1;V_1|Q) - I(\overline{W}_2;V_2|Q) < & I(\widetilde{W}_1, X_2; Y_2|\widetilde{W}_2,Q) +I(\widetilde{W}_2, X_1; Y_1|\widetilde{W}_1,Q), \\
H(U_1)+H(U_2) + I(\overline{W}_1; U_1|V_1,Q) - I(\overline{W}_2;V_2|Q) < & I(\widetilde{W}_2, X_1; Y_1|Q) + I(\widetilde{W}_1,X_2; Y_2|\widetilde{W}_2,Q),  \\
H(U_1)+H(U_2) + I(\overline{W}_2; U_2|V_2,Q) - I(\overline{W}_1;V_1|Q) < & I(\widetilde{W}_1, X_2; Y_2|Q) + I(\widetilde{W}_2,X_1; Y_1|\widetilde{W}_1,Q),  \\
2H(U_1)+H(U_2)- I(\overline{W}_1;V_1|Q) - I(\overline{W}_2;V_2|Q)  < & I(\widetilde{W}_2, X_1; Y_1|Q) + I(X_1; Y_1|\widetilde{W}_1,\widetilde{W}_2,Q) \nonumber \\
& ~~~~~ + I(\widetilde{W}_1, X_2; Y_2|\widetilde{W}_2,Q) \mbox{ and }
\end{align}
\begin{align}
H(U_1)+2H(U_2)- I(\overline{W}_1;V_1|Q) - I(\overline{W}_2;V_2|Q)  < & I(\widetilde{W}_1, X_2; Y_2|Q) + I(X_2; Y_2|\widetilde{W}_2,\widetilde{W}_1,Q) \nonumber \\
& ~~~~~ + I(\widetilde{W}_2, X_1; Y_1|\widetilde{W}_1,Q),
\end{align}
for some $p(q)$, $p(\overline{w}_1|u_1,q)$, $p(\overline{w}_2|u_2,q)$, $p(\widetilde{w}_1,x_1|q)$ and $p(\widetilde{w}_2,x_2|q)$, where the entropies and mutual information terms are evaluated using joint distribution
\begin{align}
p(q, u_1,v_1,u_2,v_2,\overline{w}_1, \overline{w}_2,\widetilde{w}_1,\widetilde{w}_2, x_1,x_2,y_1,y_2) = p(q) & p(u_1,v_1) p(\overline{w}_1|u_1,q)p(u_2,v_2) \nonumber\\
& p(\overline{w}_2|u_2,q)  p(\widetilde{w}_1, x_1|q)p(\widetilde{w}_2, x_2|q) p(y_1, y_2|x_1,x_2).\label{dist:sep}
\end{align}
\end{corollary}
\begin{proof}
Corollary \ref{c:separate} follows directly from Theorem \ref{wish} by letting $W_k=(\overline{W}_k, \widetilde{W}_k)$ and fixing the distributions as $p(w_k,x_k|u_k, q) = p(\overline{w}_k|u_k,q)p(\widetilde{w}_k, x_k|q)$, for $k=1,2$.
\end{proof}
The sufficient conditions in Corollary \ref{c:separate} are looser than those in Theorem \ref{wish}. However, it is not clear whether they are strictly looser.

\begin{rem}
In the special case of no receiver side information, i.e., $V_1=V_2=\emptyset$, by fixing $\overline{W}_1 =\overline{W}_2=\emptyset$, and defining $R_1=H(U_1)$ and $R_2=H(U_2)$, the sufficiency conditions in Corollary \ref{c:separate} boils down to the Han-Kobayashi rate region in the form expressed in \cite[Theorem 2]{Chong:2006}.
\end{rem}

We do not know whether the sufficient conditions for reliable transmission provided in Theorem \ref{wish} are too strong, leading to pessimistic results in general. However, in Section \ref{special}, we show that for some special cases, the sufficient conditions obtained through separate source and channel coding in Corollary \ref{c:separate} are also necessary, which shows that at least for certain special cases, Theorem \ref{wish} is tight.

\subsection{Deterministic Side Information} \label{sepdet}

In this subsection, we focus on the special case in which the side information sequences $V_1$ and $V_2$ are deterministic functions of the sources $U_1$ and $U_2$, respectively, i.e.,
\begin{align}
V_{k,i}=h_k(U_{k,i}), \qquad k=1,2, \quad i=1,2,\cdots \label{deterministic_condition}
\end{align}
for some deterministic functions $h_1$ and $h_2$, or equivalently we have $H(V_k|U_k)=0$ for $k=1,2$.

The main result of this subsection is that when the side information is a deterministic function of the interfering source, the source-channel separation theorem applies; that is, it is optimal to first perform source coding and encode $V_k^n$ into message $W_{ks}$, and the remaining part of $U_k^n$, denoted by $U_k^n|V_k^n$, into message $W_{kp}$, $k=1,2$, and then to transmit these messages optimally over the underlying interference channel $p(y_1,y_2|x_1,x_2)$ with side information $W_{1s}$ at Receiver 2, and side information $W_{2s}$ at Receiver 1.

\begin{figure*}
\centering
\psfrag{X1}{$X_1$}\psfrag{X2}{$X_2$}
\psfrag{Y1}{$Y_1$}\psfrag{Y2}{$Y_2$}
\psfrag{W1p}{$W_{1p}$}\psfrag{W1s}{$W_{1s}$}
\psfrag{W2p}{$W_{2p}$}\psfrag{W2s}{$W_{2s}$}
\psfrag{sW1s}{$W_{1s}$}\psfrag{sW2s}{$W_{2s}$}
\psfrag{hW1p}{$\hat{W}_{1p}$}\psfrag{hW1s}{$\hat{W}_{1s}$}
\psfrag{hW2p}{$\hat{W}_{2p}$}\psfrag{hW2s}{$\hat{W}_{2s}$}
\psfrag{Tx1}{$\mathrm{Transmitter~1}$} \psfrag{Tx2}{$\mathrm{Transmitter~2}$}
\psfrag{Rx1}{$\mathrm{Receiver~1}$} \psfrag{Rx2}{$\mathrm{Receiver~2}$}
\psfrag{pxy}{$p(y_1, y_2|x_1, x_2)$}
\includegraphics[width=5in]{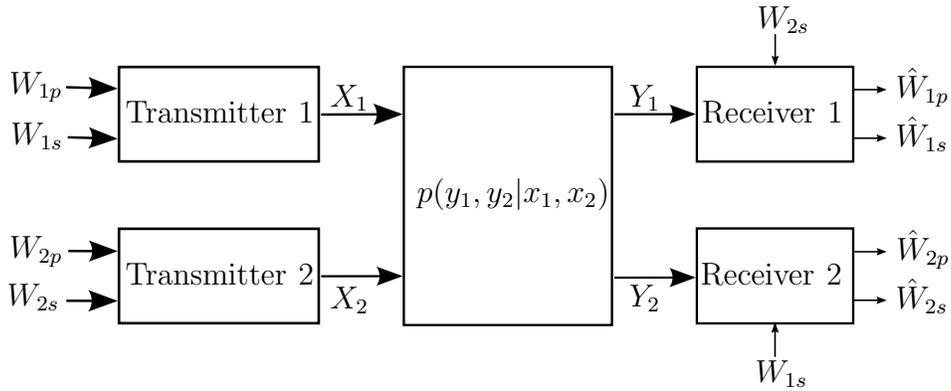}
\caption{Interference channel with message side information at the receivers.} \label{f:message_side}
\end{figure*}

First, we define the capacity region of the interference channel with message side information at the receivers (see Fig. \ref{f:message_side}). In this communication scenario, Transmitter $k$ has two messages $W_{ks}$ and $W_{kp}$, of rates $R_{ks}$ and $R_{kp}$ respectively, to transmit with negligible probability of error to Receiver $k$, $k=1,2$, while Receiver $2$ has access to  $W_{1s}$, and Receiver 1 has access to $W_{2s}$. All messages are independent. A $\left(2^{nR_{1s}}, 2^{nR_{1p}}, 2^{nR_{2s}}, 2^{nR_{2p}},n \right)$ code for this channel consists of two encoding functions,
\begin{align}
f_1^n:& \{1,2,\cdots, 2^{nR_{1s}}\} \times \{1,2,\cdots, 2^{nR_{1p}}\} \rightarrow \mathcal{X}_1^n
\end{align}
and
\begin{align}
f_2^n:& \{1,2,\cdots, 2^{nR_{2s}}\} \times \{1,2,\cdots, 2^{nR_{2p}}\} \rightarrow \mathcal{X}_2^n
\end{align}
and two decoding functions
\begin{align}
g_1^n:& \mathcal{Y}_1^n \times \{1,2,\cdots, 2^{nR_{2s}}\} \rightarrow \{1,2,\cdots, 2^{nR_{1s}}\} \times \{1,2,\cdots, 2^{nR_{1p}}\}
\end{align}
and
\begin{align}
g_2^n:& \mathcal{Y}_2^n \times \{1,2,\cdots, 2^{nR_{1s}}\} \rightarrow \{1,2,\cdots, 2^{nR_{2s}}\} \times \{1,2,\cdots, 2^{nR_{2p}}\}.
\end{align}
The average probability of error for the $\left(2^{nR_{1s}}, 2^{nR_{1p}}, 2^{nR_{2s}}, 2^{nR_{2p}},n \right)$ code is defined as
\begin{align}
P_e^n=\frac{1}{2^{n(R_{1s}+R_{1p}+R_{2s}+R_{2p})}}\sum_{w_{1s}=1}^{2^{nR_{1s}}}&\sum_{w_{1p}=1}^{2^{nR_{1p}}}\sum_{w_{2s}=1}^{2^{nR_{2s}}}\sum_{w_{2p}=1}^{2^{nR_{2p}}} \text{Pr} \{g_1^n(Y_1^n,w_{2s}) \neq (w_{1s}, w_{1p}) \nonumber\\
& \text{ or } g_2^n(Y_2^n,w_{1s}) \neq (w_{2s}, w_{2p})|(w_{1s}, w_{1p}, w_{2s}, w_{2p}) \text{ is sent} \}.
\end{align}

\begin{Def}
A rate quadruplet $(R_{1s}, R_{1p}, R_{2s}, R_{2p})$ is said to be achievable if there exists a sequence of $\left(2^{nR_{1s}}, 2^{nR_{1p}}, 2^{nR_{2s}}, 2^{nR_{2p}},n \right)$ codes for which $P_e^n \rightarrow 0$ as $n \rightarrow\infty$. The capacity region is defined as the closure of the set of achievable rate quadruplets $(R_{1s}, R_{1p}, R_{2s}, R_{2p})$, and is denoted by $\mathcal{C}_I$.
\end{Def}

In order to show the optimality of source-channel code separation, similarly to Theorem \ref{ownSideinfo},  we will use the $n$-letter characterization of $\mathcal{C}_I$ provided in the next lemma. Define $\mathcal{G}^n$ as
\begin{align}
\mathcal{G}^n=\bigg\{\left(R_{1s}, R_{1p}, R_{2s}, R_{2p} \right):  &R_{1p} \leq \frac{1}{n} I(X_1^n;Y_1^n|S_{1s}^n, S_{2s}^n),  R_{1s}+R_{1p} \leq \frac{1}{n} I(X_1^n;Y_1^n|S_{2s}^n), \nonumber \\
& R_{2p} \leq \frac{1}{n} I(X_2^n;Y_2^n|S_{1s}^n, S_{2s}^n), R_{2s}+R_{2p} \leq \frac{1}{n} I(X_2^n;Y_2^n|S_{1s}^n), \nonumber \\
& \hspace{0.65in} \text{for any } p^n(s_{1s}^n)p^n(s_{2s}^n)p^n(x_1^n|s_{1s}^n) p^n(x_2^n|s_{2s}^n) \bigg\}
\end{align}
\begin{lemma} \label{n_letter_message_region}
The capacity region of the interference channel with  message side information $W_{1s}$ at Receiver 2, and message side information $W_{2s}$ at Receiver 1 is
\begin{align}
\mathcal{C}_I=\lim_{n \rightarrow \infty} \quad \mathcal{G}^n \label{ci}
\end{align}
where the limit of the region is as defined in \cite[Theorem 5]{Shannon:1961}.
\end{lemma}
\begin{proof}
A proof of Lemma \ref{n_letter_message_region} is given in Appendix \ref{proof_n_letter_message_region}.
\end{proof}

Now that we have the $n$-letter characterization of the capacity region of interference channels with message side information at the receivers, we are ready to show that the source-channel separation theorem holds when the receivers' side information sequences are deterministic functions of the interfering sources.
\begin{theorem} \label{sepnew}
Sources $U_1$ and $U_2$ can be transmitted reliably to their respective receivers over the discrete memoryless interference channel $p(y_1,y_2|x_1,x_2)$ with side information $V_1=h_1(U_1)$ at Receiver $2$, and side information $V_2=h_2(U_2)$ at Receiver 1, if
\begin{eqnarray}
(H(V_1), H(U_1|V_1), H(V_2), H(U_2|V_2)) \in int(\mathcal{C}_I), \label{si_int}
\end{eqnarray}
where $\mathcal{C}_I$ denotes the capacity region of the interference channel with message side information at receivers.

Conversely, if $(H(V_1), H(U_1|V_1), H(V_2), H(U_2|V_2)) \notin  \mathcal{C}_I$, then sources $U_1$ and $U_2$ cannot be transmitted reliably.
\end{theorem}
\begin{proof}
A proof of Theorem \ref{sepnew} is given in Appendix \ref{proof_theorem_deterministic}.
\end{proof}
The benefits of considering the side information samples as deterministic functions of the source samples are two-fold. Firstly, the transmitters also know the side information and they can use this knowledge to minimize the amount of interference they cause. Due to this fact, we are able to achieve any point in the capacity region of the interference channel with message side information. Secondly, encoding the information of $V_k$, $k=1,2$ into the codebook at Transmitter $k$ not only helps reduce the interference at the other receiver, but also does not place any extra burden on Receiver $k$ to decode $V_k$, as $V_k$ is a deterministic function of $U_k$. This fact enables the converse proof of the source-channel separation theorem.

\subsection{Necessary and Sufficient Conditions for Reliable Transmission for a Special Case} \label{special}
In Section \ref{sepdet}, we have shown that source-channel separation is optimal when the side information is a deterministic function of the interfering source. Thus, for these cases, if the single-letter characterization of the capacity region of the corresponding interference channel with message side information, i.e., $\mathcal{C}_I$, is known, we would have necessary and sufficient conditions for reliable transmission in a single-letter form. However, a single-letter characterization of $\mathcal{C}_I$ is not known in general as it is a generalization of the capacity region of the classical interference channel.

In this subsection, we consider the class of interference channels studied in \cite{Liu_Goldsmith:2008ISIT}. We show that the Han-Kobayashi scheme is capacity-achieving for this class of interference channels \cite{Liu_Goldsmith:2008ISIT} when the receivers have message side information, and we obtain a single-letter characterization of the capacity region. Hence, we conclude that, for this class of interference channels, when the side information is a deterministic function of the interfering source, the sufficient conditions provided in Theorem \ref{wish} are also necessary, yielding a single-letter characterization of the necessary and sufficient conditions for reliable transmission. This means that the achievability result presented in Theorem \ref{wish} is tight in some special cases.

The special class of interference channels we focus on in this subsection is a class of Z-interference channels. For the Z-interference channels, $p(y_1,y_2|x_1,x_2)$ can be written as $p(y_2|x_1,x_2) \cdot p(y_1|x_1)$, i.e., the channel between $X_1$ and $Y_1$ is a single user channel characterized by $p(y_1|x_1)$. This corresponds to an interference channel in which only the second transmitter-receiver pair faces interference. In particular, the members of the class of Z-interference channels we consider satisfy the following conditions:
\begin{enumerate}
\item For any positive integer $n$, $H(Y_2^n|X_2^n=x_2^n)$, when evaluated with the distribution $\sum_{x_1^n}p(x_1^n) \break p(y_2^n|x_1^n,x_2^n)$,  is independent of $x_2^n$ for any $p(x_1^n)$. \label{shiftn}
\item Define $\tau$ as
\begin{align}
\tau=\max_{p(x_1)p(x_2)} H(Y_2). \label{definetau}
\end{align}
Then there exists a $p^*(x_2)$ such that $H(Y_2)$, when evaluated with the distribution $\sum_{x_1,x_2}p(x_1)$ $p^*(x_2)  p(y_2|x_1,x_2)$, is equal to $\tau$ for any $p(x_1)$.  \label{max}
\end{enumerate}
Please refer to \cite{Liu_Goldsmith:2008ISIT} for intuition behind these conditions and examples of Z-interference channels that satisfy these two conditions.

In the next lemma, we provide a single-letter characterization of $\mathcal{C}_I$, i.e., the capacity region of this class of Z-interference channels with message side information. Since Receiver 1 does not face interference, there is no benefit to having access to the side information $W_{2s}$. Hence, without loss of generality, we assume $R_{2s}=0$.
\begin{lemma} \label{message_region}
The capacity region of Z-interference channels satisfying Conditions 1 and 2, with message side information $W_{1s}$ at Receiver 2, is characterized by
\begin{align}
R_{1p}+R_{1s} & \leq I(X_1;Y_1), \label{messageZ1}\\
R_{2p} & \leq I(W, X_2;Y_2) \mbox{ and}\\
R_{1p}+R_{2p} & \leq I(X_1;Y_1|W)+I(W,X_2;Y_2) \label{messageZ2}
\end{align}
for some $p(w)p(x_1|w)$,
where the mutual informations and entropies are evaluated with the joint distribution of the form \[p(w,x_1,x_2, \break y_1,y_2)=p(w)p(x_1|w)p^*(x_2)p(y_1|x_1)p(y_2|x_1,x_2).\]
\end{lemma}
\begin{proof}
A proof of Lemma \ref{message_region} is given in Appendix \ref{proof_message_region}.
\end{proof}

The proof of Lemma \ref{message_region} indicates that superposition encoding and partial decoding is capacity-achieving. More specifically, the codebook at Transmitter 1 is such that the inner codebook carries the side information at Receiver 2, i.e., $W_{1s}$, and part of $W_{1p}$, and the outer codebook carries the remaining part of $W_{1p}$.

Comparing these results in the case of side information at the receiver with the traditional Z-interference channel \cite{Liu_Goldsmith:2008ISIT}, the rate of $W_{1p}$ takes the place of $W_1$, which means that the message that causes interference is reduced from $W_1$ to $W_{1p}$. Due to the fact that $W_{1s}$ is available at Receiver 2, $W_{1s}$ does not cause any interference and therefore its rate can be made as large as possible within the constraint of the capacity of the channel $p(y_1|x_1)$ depicted by (\ref{messageZ1}).

Having established the capacity region of this special class of Z-interference channels with message side information at the receiver, we next consider the joint source-channel coding problem for this channel model with the assumption that each side information sample $V_{1,i}$ is a deterministic function of the corresponding source sample $U_{1,i}$, i.e., $V_{1,i}=h_1(U_{1,i})$, for $i=1,2,\cdots$ for some deterministic function $h_1$. Since the first transmitter-receiver pair is interference-free, without loss of generality, we assume $V_2=\emptyset$.

Since source-channel separation is shown to be optimal in Theorem \ref{sepnew} for the source and side information structure under consideration, we are able to characterize necessary and sufficient conditions for the reliable transmission of the sources in the single-letter form using the capacity region characterization given in Lemma \ref{message_region}.

\begin{corollary} \label{necessary_sufficient}
For Z-interference channels satisfying Conditions 1 and 2, and side information $V_1=h_1(U_1)$ at Receiver 2,
%
%
necessary and sufficient conditions for reliable transmission are
\begin{align}
H(U_1) < & I(X_1;Y_1) \\
H(U_2) < & I(W, X_2;Y_2) \mbox{ and } \\
H(U_1|V_1)+H(U_2)  < & I(W,X_2;Y_2)+I(X_1;Y_1|W)
\end{align}
for some $p(w)p(x_1|w)$,
where the mutual informations and entropies are evaluated with $p(u_1, v_1, u_2, \break w,x_1,x_2,y_1,y_2)=p(u_1,v_1)p(u_2)p(w)p(x_1|w)p^*(x_2)p(y_1|x_1)p(y_2|x_1,x_2)$.
\end{corollary}
\begin{proof}
Corollary \ref{necessary_sufficient} follows directly from combining Theorem \ref{sepnew} and Lemma \ref{message_region}.
\end{proof}

In Corollary \ref{c:separate}, specify $V_2=\emptyset$, choose $\overline{W}_2=\widetilde{W}_2=\emptyset$, $\overline{W}_1=V_1$, $Q=\emptyset$ and $p(x_2)=p^*(x_2)$. Renaming $\widetilde{W}_1$ as $W_1$ and using Condition 2 and the fact that $H(U_1)-H(V_1) = H(U_1|V_1)$, we obtain a sufficient condition which is the same as the necessary and sufficient condition specified in Corollary \ref{necessary_sufficient}. Hence, we conclude that in this special case, the sufficient conditions described in Corollary \ref{c:separate} based on separate source and channel coding are also necessary. This shows that the conditions presented in Theorem \ref{wish} are also necessary at least in certain scenarios.

Corollary \ref{necessary_sufficient} shows how the side information $V_1=h_1(U_1)$ about the interference $U_1$ helps in reliable transmission, and determines the most efficient way of using this side information: Transmitter 1 performs a separation-based encoding scheme. It first splits its source $U_1^n$ into $V_1^n$ and a remaining part using entropy-achieving data compression techniques, and thus obtains two messages $W_{V_1^n}$ and $W_{U_1^n|V_1^n}$. Then, it further splits message $W_{U_1^n|V_1^n}$ into two parts $W_{\text{inner}}$ and $W_{\text{outer}}$, at rates $\gamma$ and $H(U_1|V_1)-\gamma$, respectively. Next, it performs superposition encoding, transmitting $W_{V_1^n}$ and $W_{\text{inner}}$ through the inner code at rate $H(V_1)+\gamma$, and $W_{\text{outer}}$ through the outer code at rate $H(U_1|V_1)-\gamma$. Transmitter 2 performs separation-based source-channel coding, first mapping $U_2^n$ into a message $W_2$ and then mapping $W_2$ into a codeword of an i.i.d. codebook generated with distribution $p^*(x_2)$. Receiver 1 decodes both the inner and the outer codes. Receiver 2 knows the side information $V_1^n$ and hence sees an inner codebook at an effective rate of $\gamma$ only. It decodes the inner codeword and the codeword of Transmitter 2 jointly using the received signal and the available side information about the interference.

The intuition obtained from the special case derived in this subsection is that one should put as much information as possible about the side information within the inner codebook, in order to minimize the impact of interference when the side information about the interference is available at the receiver.

\section{Z-Channel with Degraded Message Sets} \label{s:degmsg_Z}

The result in (\ref{messageZ1})-(\ref{messageZ2}) is directly related to the capacity region of the Z-channel with degraded message sets, based on the intuition gained from the proof of Theorem 3 in \cite{Kramer:2007}. The intuition in \cite{Kramer:2007} is that when the receiver has some side information about the undesired message, we can set up a new scenario in which the receiver does not have access to the side information, and is required to decode it. Then, when we remove the rate constraint associated with decoding of the side information at the receiver in the capacity region of the new scenario, we get the capacity results of the original scenario. Therefore, the solution given in (\ref{messageZ1})-(\ref{messageZ2}) resembles the solution of the following problem.

The channel is described by two transition probabilities $p(y_1|x_1)$ and $p(y_2|x_1,x_2)$, and satisfies both Conditions 1 and 2. There are three independent messages $W_{1c}$, $W_{1p}$ and $W_2$. Transmitter 1 has messages $W_{1c}$ and $W_{1p}$ and Transmitter 2 has message $W_2$. $W_{1c}$ needs to be decoded at both receivers, while $W_{1p}$ and $W_2$ need to be decoded only at Receiver 1 and Receiver 2, respectively.

This channel model includes the Z-interference channel as a special case, when the rate of $W_{1c}$ is zero. Compared to the definition of the Z-channel in \cite{Vishwanath:2003Globecom}, $W_{1c}$ is not only intended for Receiver 2, but also for Receiver 1. Therefore, we call this channel model as the \emph{Z-channel with degraded message sets}.

Then the capacity region for the Z-channel satisfying Conditions 1 and 2, with degraded message sets can be characterized as follows:
\begin{align}
R_{1p} & \leq I(X_1;Y_1|W)+\gamma, \\
R_{1c}+R_{1p} & \leq I(X_1;Y_1), \\
R_{1c} & \leq I(W;Y_2|X_2)-\gamma \mbox{ and }\\
R_2+R_{1c} & \leq \tau-H(Y_2|W,X_2)-\gamma,
\end{align}
for some $p(w)p(x_1|w)$ and $\gamma \geq 0$ where the mutual informations and entropies are evaluated using $
p(w, x_1, x_2, y_1, y_2)=p(w)p(x_1|w)p^*(x_2)p(y_1|x_1)p(y_2|x_1,x_2)
$.
The proof of this result follows from arguments very similar to those used in the scenario of message side information at the receiver considered in Lemma \ref{message_region}.

\section{Conclusions} \label{conclusions}

We have studied the problem of joint source-channel coding in interference channels with correlated receiver side information. In the case when the receiver side information is correlated with its desired source, we have shown that separate design of source and channel codes is optimal. In order to minimize the interference to the other transmitter-receiver pair, the transmitters should transmit only the part of their sources that is not already known by their corresponding receivers.

For the case in which the receiver side information is correlated with the interfering source, we have provided sufficient conditions for reliable transmission by proposing a joint source-channel coding scheme based on the idea of superposition encoding and partial decoding of Han and Kobayashi. As a special case, we have focused on the scenario in which the side information at the receiver is a deterministic function of the interfering source, and we have shown that source-channel separation is optimal for this situation as well. In both cases for which the optimality of source-channel separation is established, we have used the $n$-letter expression for the capacity region as a single-letter expression is not available in general.

Finally, for a class of Z-interference channels for which superposition encoding and partial decoding is optimal in the absence of receiver side information, when the receiver facing interference has access to a deterministic function of the interfering source, we have shown that the provided sufficient conditions are also necessary. Hence, the sufficient conditions are tight at least in some special cases.

\appendices

\section{Proof of Theorem \ref{ownSideinfo}} \label{proof_ownSideinfo}
The achievability part of the proof is straightforward. If (\ref{ownSideagain}) holds, then there exists a rate pair $(R_1, R_2)$ in the interior of $\mathcal{C}$ such that $H(U_k|V_k) \leq R_k$ for $k=1,2$. Each transmitter compresses its source with respect to the side information at its own receiver. This can be done at rate $R_k$ due to the Slepian-Wolf theorem. Then the compressed bits can be transmitted reliably over the channel since $(R_1,R_2)$ is in the capacity region of the underlying interference channel.

To prove the converse, we first provide an infinite letter expression for the capacity region of the interference channel given in \cite{Ahlswede:1971}. We define
\begin{align}
E_n \triangleq \bigg\{ &\left(\frac{1}{n}I(X_1^n;Y_1^n), \frac{1}{n}I(X_2^n;Y_2^n) \right): p(x_1^n, x_2^n) = p(x_1^n)p(x_2^n) \bigg\}.
\end{align}
Then
\begin{eqnarray}
\mathcal{C} = \lim_{n \rightarrow \infty} E_n, \label{cap_region}
\end{eqnarray}
where the limit is defined as in \cite[Theorem 5]{Shannon:1961}. $\mathcal{C}$ is a closed convex set in the Euclidean plane.

From Fano's inequality \cite{Cover:book}, we have, for $k=1,2$,
\begin{eqnarray}
H(U_k^n|\hat{U}_k^n) &\leq & n \delta(P_e^n),
\end{eqnarray}
where $\delta(x)$ is a non-negative function approaching zero as $x \rightarrow 0$.

Next, we write the following chain of inequalities:
\begin{eqnarray}
\frac{1}{n} I(X_1^n ; Y_1^n) &\geq&  \frac{1}{n} I(U_1^n ; Y_1^n) \label{ach_eq1} \\
        &=&  \frac{1}{n} I(U_1^n, V_1^n ; Y_1^n) \label{ach_eq2} \\
        &\geq&  \frac{1}{n} I(U_1^n ; Y_1^n | V_1^n) \label{ach_eq3} \\
        &=&  \frac{1}{n} \left[H(U_1^n | V_1^n) - H(U_1^n | V_1^n, Y_1^n) \right] \label{ach_eq4} \\
        &\geq& H(U_1|V_1) - H(U_1^n | \hat{U}_1^n)  \label{ach_eq5} \\
        &\geq& H(U_1|V_1) - \delta(P_e^n) \label{ach_eq6}
\end{eqnarray}
where (\ref{ach_eq1}) follows since $U_1^n \rightarrow X_1^n \rightarrow Y_1^n$ form a Markov chain, similarly (\ref{ach_eq2}) follows since $V_1^n \rightarrow U_1^n \rightarrow Y_1^n$ form a Markov chain, and finally (\ref{ach_eq6}) follows from Fano's inequality. Similarly, we can also show
\begin{eqnarray}
\frac{1}{n} I(X_2^n ; Y_2^n) &\geq&  H(U_2|V_2) - \delta(P_e^n),
\end{eqnarray}
where the joint probability distribution factors as $p(x_1^n)$$p(x_2^n)$.

From the capacity region given in (\ref{cap_region}), we see that $(H(U_1|V_1) - \delta(P_e^n), H(U_2|V_2) - \delta(P_e^n)) \in \mathcal{C}$ for all $n$. Then, since $\delta(P_e^n) \rightarrow 0$ as $n \rightarrow \infty$, and from the compactness of the capacity region, we can conclude that $P_e^n \rightarrow 0$ implies that $(H(U_1|V_1), H(U_2|V_2)) \in \mathcal{C}$. This completes the proof.

\section{Proof of Theorem \ref{wish}} \label{proofsufficient}

We first briefly review the notions of types and strong typicality that will be used in the proof. Given a distribution $p(x)$, the type $P_{x^n}$ of an $n$-tuple $x^n$ is the empirical distribution
\[P_{x^n} = \frac{1}{n} N(a|x^n)\]
where $N(a|x^n)$ is the number of occurrences of the letter $a$ in $x^n$. The set of all $n$-tuples $x^n$ with type $Q$ is called the type class $Q$ and is denoted by $T^n(Q)$. The set of $\delta$-strongly typical $n$-tuples according to $p(x)$ is denoted by $T_{\epsilon}^n(X)$ and is defined by
\[T_{\epsilon}^n(X) = \left\{ x^n \in \mathcal{X}^n : \left| \frac{1}{n} N(a|x^n)-p(a)\right|\leq \delta,  \forall a \in \mathcal{X} \mbox{ and } N(a|x^n)=0 \mbox{ whenever } p(x)=0   \right\}. \]

The definitions of type and strong typicality can be extended to joint and conditional distributions in a similar manner \cite{Csiszar:book}. We have
\begin{eqnarray}\label{type0}
\left|\frac{1}{n}\log |T_{\epsilon}^n(X)| - H(X) \right| \leq \delta
\end{eqnarray}
for sufficiently large $n$. Given a joint distribution $p(x,y)$, if $(x^n, y^n) \sim p^n(x) p^n(y)$, where $p^n(x)$ and $p^n(y)$ are $n$-fold products of the marginals $p(x)$ and $p(y)$, then
\begin{eqnarray}\label{type1}
\mathrm{Pr} \{(x^n, y^n) \in T_{\epsilon}^n(XY) \} \leq 2^{-n(I(X;Y)-3\delta)}.
\end{eqnarray}

Now, we start the achievability proof. Fix a joint distribution as in (\ref{dis}), where $p(u_1,v_1)$, $p(u_2,v_2)$, $p(y_1,y_2|x_1,x_2)$ are given while we are free to choose $p(q)$, $p(w_1, x_1|u_1,q)$ and  $ \break p(w_2, x_2|u_2,q)$.

\emph{Codebook generation}: First, generate one random $n$-sequence $q^n$ in an i.i.d. fashion according to $p(q)$.

Next, for Transmitter 1, generate a codebook of size $L_1$ with $\frac{1}{n} \log L_1  > I(U_1;W_1|Q)$, in which the codewords are generated i.i.d. with distribution $p(w_1|q)$. This codebook is denoted by $\mathcal{C}_w^1$.

For each possible source output $u_1^n$, count the number of codewords in $\mathcal{C}_w^1$ that are jointly typical with $u_1^n$. If there are at least $L_1 2^{-nI(U_1;W_1|Q)-2n \epsilon}$ codewords in $\mathcal{C}_w^1$ jointly typical with $u_1^n$, choose one uniformly at random, and call it $w_1^n(u_1^n)$. If there are fewer than $L_1 2^{-nI(U_1;W_1|Q)-2n \epsilon}$ codewords of $\mathcal{C}_w^1$ jointly typical with $u_1^n$, randomly choose one codeword from $\mathcal{C}_w^1$ to be $w_1^n(u_1^n)$. The reason why we require the number of codewords jointly typical with $u_1^n$ to be large is to benefit the probability of error calculation later on in the proof. In a similar fashion, we generate $\mathcal{C}_w^2$.

Define $F(u_1^n, u_2^n)$ as the event that the number of $w_1^n \in \mathcal{C}_w^1$ jointly typical with $u_1^n$ is larger than $L_1 2^{-nI(U_1;W_1|Q)-2n \epsilon}$ and the number of $w_2^n \in \mathcal{C}_w^2$ jointly typical with $u_2^n$ is larger than $L_2 2^{-nI(U_2;W_2|Q)-2n \epsilon}$. Next, we will show that
\begin{align}
\text{Pr}\{F^c(U_1^n, U_2^n)\} \leq 3\epsilon, \label{rate_distortion}
\end{align}
 where ``$c$'' denotes the complement.

For each $(q^n, u_1^n, u_2^n) \in T_\epsilon^n(Q U_1 U_2)$, define the random variable $\nu(i,u_1^n)$ as follows: $\nu(i,u_1^n)$ is $1$ if the $i$-th codeword of $\mathcal{C}_w^1$ is jointly typical with $u_1^n$ and $0$ otherwise. Then,
\begin{align}
2^{-nI(U_1;W_1|Q)-n \epsilon} \leq \mathbf{E}[\nu(i,u_1^n)|q^n]&=\text{Pr}\{\nu(i,u_1^n)=1|q^n\} \leq 2^{-nI(U_1;W_1|Q)+n \epsilon} \label{single1}\\
\mathbf{V}[\nu(i,u_1^n)|q^n] &\leq \mathbf{E}^2[\nu(i,u_1^n)|q^n] \leq \mathbf{E} [\nu(i,u_1^n)] \label{single2}
\end{align}
where $\mathbf{E}$ and $\mathbf{V}$ denote the expectation and variance, respectively. Further define random variable $N(u_1^n)$ as the number of codewords in $\mathcal{C}_w^1$ that are jointly typical with $u_1^n$, i.e.,
\begin{align}
N(u_1^n)=\sum_{i=1}^{L_1} \nu(i,u_1^n).
\end{align}
Then, from (\ref{single1}) and (\ref{single2}), we have
\begin{align}
L_1 2^{-nI(U_1;W_1|Q)-n \epsilon} \leq \mathbf{E}[N(u_1^n)|q^n]&=\sum_{i=1}^{L_1} \mathbf{E}[\nu(i,u_1^n)|q^n] \leq L_1 2^{-nI(U_1;W_1|Q)+n\epsilon}  \label{bound1}\\
\mathbf{V}[N(u_1^n)|q^n]&=\sum_{i=1}^{L_1} \mathbf{V}[\nu(i,u_1^n)|q^n] \leq \mathbf{E}[N(u_1^n)|q^n].\label{bound2}
\end{align}
Hence, we have
\begin{align}
\text{Pr}\big\{N(u_1^n) \leq & L_1 2^{-nI(U_1;W_1|Q)-2n \epsilon}|q^n\big\}\nonumber\\
& = \text{Pr}\left\{\mathbf{E}[N(u_1^n)|q^n]-N(u_1^n) \geq  \mathbf{E}[N(u_1^n)|q^n]-L_1 2^{-nI(U_1;W_1|Q)-2n \epsilon}|q^n\right\}\\
& \leq \text{Pr}\left\{\mathbf{E}[N(u_1^n)|q^n]-N(u_1^n) \geq  L_1 2^{-nI(U_1;W_1|Q)-n\epsilon} -L_1 2^{-nI(U_1;W_1|Q)-2n \epsilon}|q^n\right\} \label{larger1}\\
& \leq \text{Pr}\left\{\big|\mathbf{E}[N(u_1^n)|q^n]-N(u_1^n) \big| \geq  L_1 2^{-nI(U_1;W_1|Q)-n\epsilon} -L_1 2^{-nI(U_1;W_1|Q)-2n \epsilon}\big|q^n\right\}\\
& \leq \frac{\mathbf{V}[N(u_1^n)|q^n]}{\left(L_1 2^{-nI(U_1;W_1|Q)-n\epsilon} -L_1 2^{-nI(U_1;W_1|Q)-2n \epsilon} \right)^2} \label{chebyshev}\\
& \leq \frac{\mathbf{E}[N(u_1^n)|q^n]}{\left(L_1 2^{-nI(U_1;W_1|Q)-n\epsilon} -L_1 2^{-nI(U_1;W_1|Q)-2n \epsilon} \right)^2} \label{larger2}\\
& \leq \frac{L_1 2^{-nI(U_1;W_1|Q)+n\epsilon}}{\left(L_1 2^{-nI(U_1;W_1|Q)-n\epsilon} -L_1 2^{-nI(U_1;W_1|Q)-2n \epsilon} \right)^2} \label{larger3}\\
& \leq \epsilon \label{larger4}
\end{align}
where (\ref{larger1}) and (\ref{larger3}) follows from (\ref{bound1}),  (\ref{chebyshev}) follows from Chebyshev's inequality, (\ref{larger2}) follows from (\ref{bound2}), and (\ref{larger4}) is true when $n$ is large enough. The same analysis applies for $u_2^n$.

Hence, we have proved that
\begin{align}
\text{Pr}\{F^c(u_1^n, u_2^n)|q^n\} & = \text{Pr}\left\{N(u_1^n) \leq L_1 2^{-nI(U_1;W_1|Q)-2n \epsilon} \text{ or } N(u_2^n) \leq L_2 2^{-nI(U_2;W_2|Q)-2n \epsilon}|q^n\right\}  \nonumber \\
& \leq  2\epsilon \label{uselaternan}
\end{align}
for all $(q^n, u_1^n, u_2^n) \in T_\epsilon^n(Q U_1 U_2)$ and all sufficiently large $n$. This means that
\begin{align}
\text{Pr}\{F^c(U_1^n, U_2^n)\} = & \sum_{q^n, u_1^n, u_2^n} \text{Pr}\{F^c(U_1^n, U_2^n)|(U_1^n, U_2^n, Q^n)=(u_1^n, u_2^n, q^n)\}  \nonumber \\
& ~~~~~~~~~~ \cdot \text{Pr}\{(U_1^n, U_2^n, Q^n)=(u_1^n, u_2^n, q^n)\} \\
= & \sum_{(q^n, u_1^n, u_2^n) \in T_\epsilon^n(QU_1 U_2)} \text{Pr}\{F^c(U_1^n, U_2^n)|(U_1^n, U_2^n, Q^n)=(u_1^n, u_2^n, q^n)\}  \nonumber \\
& ~~~~~~~~~~ \cdot \text{Pr}\{(U_1^n, U_2^n, Q^n)=(u_1^n, u_2^n, q^n)\} \nonumber \\
& + \sum_{(q^n, u_1^n, u_2^n) \notin T_\epsilon^n(QU_1 U_2)} \text{Pr}\{F^c(U_1^n, U_2^n)|(U_1^n, U_2^n, Q^n)=(u_1^n, u_2^n, q^n)\}  \nonumber \\
& ~~~~~~~~~~ \cdot \text{Pr}\{(U_1^n, U_2^n, Q^n)=(u_1^n, u_2^n, q^n)\} \\
\leq &  2\epsilon+\text{Pr}\{(Q^n U_1^n, U_2^n) \notin T_\epsilon^n (Q U_1 U_2)\} \label{usebeforenan} \\
\leq &  3 \epsilon \label{rate_distortion}
\end{align}
where (\ref{usebeforenan}) follows from (\ref{uselaternan}), and (\ref{rate_distortion}) follows when $n$ is large enough from the asymptotic equipartition property (AEP) \cite{Cover:book}.

This means that with large probability, the number of sequences jointly typical with $U_1^n$ and $U_2^n$ in codebooks $\mathcal{C}_w^1$ and $\mathcal{C}_w^2$ are larger than $L_1 2^{-nI(U_1;W_1|Q)-2n\epsilon}$ and $L_2 2^{-nI(U_2;W_2|Q)-2n\epsilon}$, respectively. This fact will be used in the probability of error calculation.

\emph{Codebook generation}: For each possible $u_1^n$ sequence, generate one $x_1^n$ sequence in an i.i.d. fashion, conditioned on $w_1^n(u_1^n)$, $u_1^n$ and $q^n$, according to $p(x_1|u_1,w_1,q)$. This $x_1^n$ sequence is denoted by $x_1^n(u_1^n, w_1^n(u_1^n))$. The collection of all $x_1^n$ sequences will be denoted as the codebook $\mathcal{C}_x^1$.  Similarly, we generate the codebook $\mathcal{C}_x^2$.


\emph{Encoding}:
When Transmitter 1 observes the sequence $u_1^n$, it transmits $x_1^n(u_1^n, w_1^n(u_1^n))$. Similarly for Transmitter 2.

\emph{Decoding}:
Receiver 1 finds the unique pair $(u_1^n, w_2^n)$, $u_1^n \in \mathcal{U}_1^n$, $w_2^n \in \mathcal{C}_w^2$, such that $(u_1^n, w_1^n(u_1^n)$, $x_1^n(u_1^n,w_1^n(u_1^n))$, $w_2^n$, $y_1^n, v_2^n)$ are jointly typical and declares the first component of the pair as the transmitted source. If there are more than one pair, and the first component of the pairs are the same, then the decoder declares the transmitted source to be the first component. If there are more than one pair, and the first component of the pairs are not the same, an error is declared. Also, if no such pair exists, an error is declared. Similarly for Receiver 2.

\emph{Probability of error calculation}:
Denote by $E(u_1^n, w_2^n)$ the event $(u_1^n, w_1^n(u_1^n), X_1^n(u_1^n,w_1^n(u_1^n)), \break w_2^n, Y_1^n, V_2^n) \in T_\epsilon^n(U_1 W_1 X_1 W_2 Y_1 V_2|q^n)$ for $(u_1^n, w_2^n) \in \mathcal{U}_1^n \times \mathcal{C}_w^2$. Further denote by $G(u_1^n, u_2^n)$ the event $(u_1^n, u_2^n, w_1^n(u_1^n), w_2^n(u_2^n)) \in T_\epsilon^n(U_1 U_2 W_1 W_2|q^n)$.

Then, the probability of error at Receiver 1 conditioned on $Q^n=q^n$, denoted by $P_e^1$, is given by
\begin{align}
\text{Pr} \left\{E^c \right. & \left. (U_1^n, w_2^n(U_2^n) )\text{ or }\bigcup_{(u_1^n, w_2^n): u_1^n \neq U_1^n} E(u_1^n, w_2^n) \right\} \\
\leq & \text{Pr} \Bigg\{E^c(U_1^n, w_2^n(U_2^n) ) \text{ or }F^c(U_1^n, U_2^n) \text{ or } G^c(U_1^n, U_2^n) \text{ or } \bigcup_{(u_1^n, w_2^n):u_1^n \neq U_1^n} E(u_1^n, w_2^n) \Bigg\}\\
\leq & \text{Pr} \left\{E^c(U_1^n, w_2^n(U_2^n) ) \text{ or }F^c(U_1^n, U_2^n) \text{ or } G^c(U_1^n, U_2^n) \right\} \nonumber\\
&~~~~ + \text{Pr} \left\{\bigcup_{(u_1^n, w_2^n):u_1^n \neq U_1^n} E(u_1^n, w_2^n) \bigg| E\cap F\cap G \right\}\\
\leq & \text{Pr}\left\{F^c(U_1^n, U_2^n)  \right\}+ \text{Pr} \left\{G^c(U_1^n, U_2^n)|F \right\}+\text{Pr} \left\{E^c(U_1^n, w_2^n(U_2^n) ) |F \cap G \right\} \nonumber\\
&~~~~ + \mathbf{E}\left\{ \sum_{(u_1^n, w_2^n):u_1^n \neq U_1^n} \text{Pr} \left\{ E(u_1^n, w_2^n) |E \cap F \cap G\right\} \right\}, \label{firstbreakdown}
\end{align}
where we have used the short-hand $E$, $F$ and $G$ to denote events $E(U_1^n, w_2^n(U_2^n) )$, $F(U_1^n, U_2^n)$ and $G(U_1^n, U_2^n)$, respectively.

The first term in (\ref{firstbreakdown}) is bounded by $3\epsilon$ as shown by (\ref{rate_distortion}). From the achievability results of multi-terminal rate-distortion theory \cite{Berger:book}, the second term in (\ref{firstbreakdown}) is bounded by $\epsilon$ for sufficiently large $n$. The third term in (\ref{firstbreakdown}) is bounded by $\epsilon$ for sufficiently large $n$ based on the AEP \cite{Cover:book}. Hence, from now on, we will concentrate on the fourth term in (\ref{firstbreakdown}).

The fourth term in (\ref{firstbreakdown}) may be upper bounded by the sum of the following four terms, which will be denoted by $A_1, A_2, A_3$, and $A_4$, respectively:
\begin{align}
A_1 \overset{\triangle}{=} &\mathbf{E}\left\{ \sum_{\scriptsize{\shortstack{$u_1^n \neq U_1^n$\\$w_1^n(u_1^n) \neq w_1^n(U_1^n)$}}} \text{Pr} \left\{ E(u_1^n, w_2^n(U_2^n)) |E\cap F \cap G\right\} \right\} \\
A_2 \overset{\triangle}{=}&\mathbf{E}\left\{ \sum_{\scriptsize{\shortstack{$u_1^n \neq U_1^n$\\$w_1^n(u_1^n) \neq w_1^n(U_1^n)$\\$w_2^n \neq w_2^n(U_2^n)$}}} \text{Pr} \left\{ E(u_1^n, w_2^n) |E\cap F \cap G\right\} \right\}
\end{align}
\begin{align}
A_3 \overset{\triangle}{=}&\mathbf{E}\left\{ \sum_{\scriptsize{\shortstack{$u_1^n \neq U_1^n$\\$w_1^n(u_1^n) = w_1^n(U_1^n)$}}} \text{Pr} \left\{ E(u_1^n, w_2^n(U_2^n)) |E \cap F \cap G\right\} \right\}
 \end{align}
 and
 \begin{align}
A_4 \overset{\triangle}{=}& \mathbf{E}\left\{ \sum_{\scriptsize{\shortstack{$u_1^n \neq U_1^n$\\$w_1^n(u_1^n)= w_1^n(U_1^n)$\\$w_2^n \neq w_2^n(U_2^n)$}}} \text{Pr} \left\{ E(u_1^n, w_2^n)| E \cap F \cap G\right\} \right\}.
\end{align}
First, we upper bound $A_1$. Define the set
\begin{align}
\mathcal{B}_1=\{&u_1^n \in \mathcal{U}_1^n: u_1^n \neq U_1^n, w_1^n(u_1^n) \neq w_1^n(U_1^n), (u_1^n, w_1^n(u_1^n)) \in T_\epsilon^n(U_1 W_1|Y_1^n V_2^n w_2^n(U_2^n) q^n)\}.
\end{align}
Then, we have
\begin{align}
\mathbf{E} \left\{ |\mathcal{B}_1|\big|E \cap F \cap G \right\} \leq &2^{nH(U_1|Y_1, V_2, W_2, Q)+n\epsilon} 2^{nH(W_1|U_1, Y_1, V_2, W_2, Q)+n\epsilon} 2^{-nH(W_1|U_1,Q)+n \epsilon}.
\end{align}
%
%
%
Hence, we may write
\begin{align}
A_1 & = \mathbf{E} \left\{\sum_{u_1^n \in \mathcal{B}_1} \text{Pr} \left\{ E(u_1^n, w_2^n(U_2^n)) |E\cap F \cap G\right\} \right\} \label{error1start}\\
& \leq \mathbf{E} \left\{|\mathcal{B}_1|\max_{u_1^n \in \mathcal{B}_1} \text{Pr} \left\{ E(u_1^n, w_2^n(U_2^n)) |E\cap F \cap G\right\} \right\}\\
& = \mathbf{E} \bigg\{|\mathcal{B}_1|\max_{u_1^n \in \mathcal{B}_1} \text{Pr} \{ X_1^n(u_1^n, w_1^n(u_1^n)) \in  T_\epsilon^n(X_1|u_1^n w_1^n(u_1^n) w_2^n(U_2^n) Y_1^n V_2^n q^n) |E\cap F \cap G\} \bigg\}\\
& \leq \mathbf{E} \bigg\{|\mathcal{B}_1|\max_{u_1^n \in \mathcal{B}_1} 2^{nH(X_1|U_1,W_1, W_2, Y_1, V_2, Q)+n\epsilon} 2^{-nH(X_1|U_1,W_1,Q)+n \epsilon} \big| E \cap F \cap G  \bigg\}\\
& \leq 2^{nH(U_1)} 2^{-n I(U_1,W_1,X_1;Y_1, V_2|W_2, Q)+5n \epsilon} \label{error1}\\
& \leq 2^{nH(U_1)} 2^{-n I(X_1;Y_1, V_2|W_2, Q)+5n \epsilon} \label{MarkovA1}
\end{align}
where (\ref{MarkovA1}) follows because the distribution in (\ref{dis}) satisfies the Markov chain relationship $(U_1, W_1) \rightarrow (X_1, W_2, Q) \rightarrow (V_2, Y_1)$.
Next, we upper bound $A_2$. Define the set
\begin{align}
\mathcal{B}_2=\{u_1^n \in \mathcal{U}_1^n, w_2^n \in \mathcal{C}_w^2: u_1^n \neq U_1^n,& w_1^n(u_1^n) \neq w_1^n(U_1^n),
w_2^n \neq w_2^n(U_2^n), \nonumber\\
&(u_1^n, w_1^n(u_1^n), w_2^n) \in T_\epsilon^n(U_1 W_1 W_2|Y_1^n V_2^n q^n) \}.
\end{align}
Then, we have
\begin{align}
\mathbf{E} \{|\mathcal{B}_2| \} \leq &2^{nH(W_2|Y_1,V_2, Q)+n\epsilon} 2^{-nH(W_2|Q)+n\epsilon} (L_2-1)\nonumber\\
& 2^{n H(U_1|W_2,Y_1,V_2,Q)+n \epsilon}
 2^{n H(W_1|U_1,W_2,Y_1, V_2, Q)+n \epsilon}2^{-n H(W_1|U_1, Q)+n \epsilon}.
\end{align}
%
%
%
Similarly to (\ref{error1start})-(\ref{error1}), we may write
\begin{align}
A_2 &=\mathbf{E}\left\{ \sum_{(u_1^n, w_2^n) \in \mathcal{B}_2} \text{Pr} \left\{ E(u_1^n, w_2^n) |E\cap F \cap G\right\} \right\}\\
& \leq 2^{nH(U_1)} L_2 2^{-nI(U_1,W_1, X_1, W_2;V_2,Y_1|Q)+7 n \epsilon}\\
&=2^{nH(U_1)} L_2 2^{-nI(X_1, W_2;V_2,Y_1|Q)+7 n \epsilon} \label{MarkovA2}
\end{align}
where (\ref{MarkovA2}) follows from the same reason as (\ref{MarkovA1}).
Next, we upper bound $A_3$. Define the set
\begin{align}
\mathcal{B}_3=\{u_1^n \in& \mathcal{U}_1^n: u_1^n \neq U_1^n, w_1^n(u_1^n)=w_1^n(U_1^n),u_1^n \in T_\epsilon^n(U_1|w_1^n(U_1^n) Y_1^n V_2^n w_2^n(U_2^n) q^n) \}.
\end{align}
Then, we have
\begin{align}
\mathbf{E} \left\{|\mathcal{B}_3| \big| E \cap F \cap G \right\} \leq 2^{nH(U_1|W_1,Y_1, V_2, W_2, Q)+n\epsilon} \frac{1}{2^{-nI(U_1;W_1|Q)-2n \epsilon} L_1}
\end{align}
which follows from the fact that we always choose randomly from at least $L_1 2^{n I(U_1;W_1|Q)-2n\epsilon}$ choices to get $w_1^n(u_1^n)$.
Similarly to (\ref{error1start})-(\ref{error1}), we may write
\begin{align}
A_3& =\mathbf{E} \left\{\sum_{u_1^n \in \mathcal{B}_3} \text{Pr} \left\{ E(u_1^n, w_2^n(U_2^n)) |E\cap F \cap G\right\} \right\}\\
& \leq \frac{2^{nH(U_1)}}{L_1} 2^{-n I(U_1,X_1;Y_1, V_2|W_1, W_2, Q)+5 n \epsilon}\\
& \leq \frac{2^{nH(U_1)}}{L_1} 2^{-n I(X_1;Y_1, V_2|W_1, W_2, Q)+5 n \epsilon} \label{MarkovA3}
\end{align}
where (\ref{MarkovA3}) follows from the same reason as (\ref{MarkovA1}).
Finally, we upper bound $A_4$. Define the set
\begin{align}
\mathcal{B}_4=\{u_1^n \in \mathcal{U}_1^n, w_2^n \in \mathcal{C}_w^2: u_1^n \neq U_1^n, &w_1^n(u_1^n)=w_1^n(U_1^n),
w_2^n \neq w_2^n(U_2^n), \nonumber\\
& (u_1^n, w_2^n) \in T_\epsilon^n(U_1W_2|w_1^n(U_1^n) Y_1^n V_2^n q^n) \}.
\end{align}
Then, we have
\begin{align}
\mathbf{E} \left\{|\mathcal{B}_4| \big| E \cap F \cap G \right\}  \leq & 2^{nH(W_2|Y_1,V_2, W_1, Q)+n\epsilon} 2^{-nH(W_2|Q)+n\epsilon} (L_2-1) \nonumber\\
& 2^{n H(U_1|W_1,W_2,Y_1,V_2,Q)+n\epsilon} \frac{1}{2^{-nI(U_1;W_1|Q) -2n \epsilon}L_1}.
\end{align}
Similarly to (\ref{error1start})-(\ref{error1}), we may write
\begin{align}
A_4 &= \mathbf{E} \left\{ \sum_{(u_1^n,w_2^n) \in \mathcal{B}_4} \left[E(u_1^n, w_2^n)| E \cap F \cap G\right] \right\}\\
& \leq \frac{L_2}{L_1} 2^{nH(U_1)} 2^{-n I(U_1, X_1, W_2;Y_1, V_2|W_1, Q)+7 n \epsilon}\\
& \leq \frac{L_2}{L_1} 2^{nH(U_1)} 2^{-n I( X_1, W_2;Y_1, V_2|W_1, Q)+7 n \epsilon} \label{MarkovA4}
\end{align}
where (\ref{MarkovA4}) follows from the same reason as (\ref{MarkovA1}).

We have similar probability of error calculations at Receiver 2. Since
\begin{align}
P_e^n \leq \mathbf{E}_{Q^n}[P_e^1+P_e^2],
\end{align}
for this achievability scheme, as long as the following equations are satisfied,
\begin{align}
H(U_1) & \leq I(X_1;V_2,Y_1|W_2, Q), \label{suf1}\\
H(U_1)-\log L_1 & \leq I(X_1;V_2,Y_1|W_1,W_2, Q),  \\
H(U_1)+\log L_2 & \leq I(W_2, X_1;V_2,Y_1|Q), \nonumber\\
H(U_1)+ \log L_2- \log L_1 & \leq I(W_2, X_1;V_2, Y_1|W_1, Q), \\
H(U_2) & \leq I(X_2;V_1,Y_2|W_1, Q), \nonumber\\
H(U_2)-\log L_2 & \leq I(X_2;V_1,Y_2|W_1,W_2, Q),\\
H(U_2)+\log L_1 & \leq I(W_1, X_2;V_1,Y_2| Q), \nonumber\\
H(U_2)+\log L_1- \log L_2 & \leq I(W_1, X_2;V_1, Y_2|W_2, Q), \\
\log L_1 &\geq I(U_1;W_1|Q) \mbox{ and }\\
\log L_2 &\geq I(U_2;W_2|Q), \label{suf2}
\end{align}
for some $p(q)$, $p(w_1,x_1|u_1,q)$, and $p(w_2,x_2|u_2,q)$,
the probability of error is arbitrarily small for sufficiently large $n$.

By Fourier-Motzkin elimination, we obtain the sufficient conditions given in Theorem \ref{wish}.

\section{Proof of Lemma \ref{n_letter_message_region}} \label{proof_n_letter_message_region}

We first start with the proof of achievability. Fix distributions $p(s_{1s})$, $p(x_1|s_{1s})$, $p(s_{2s})$ and $p(x_2|s_{2s})$. For codebook at Transmitter $k$, $k=1,2$, we generate an inner codebook of $2^{NR_{ks}}$ i.i.d. codewords of length $N$ with probability $\prod_{i=1}^N p(s_{ks,i})$. Then, for each codeword of the inner codebook, we generate an outer codebook of $2^{NR_{kp}}$ i.i.d. codewords of length $N$ with probability $\prod_{i=1}^N p(x_{k,i}|s_{ks,i})$. For $W_{ks}=w_{ks}$ and $W_{kp}=w_{kp}$, Transmitter $k$ sends the $w_{kp}$-th codeword of the $w_{ks}$-th outer codebook. For decoding, Receiver 1 finds the codeword in all possible outer codebooks that is jointly typical with the received sequence and the $w_{2s}$-th codeword of the inner codebook of Transmitter 2. Similarly for Receiver 2. The probability of error analysis follows from standard arguments \cite{Cover:book}, and we can show that the probability of error can be driven to zero as $N \rightarrow \infty$, as long as the rates satisfy the following conditions:
\begin{align}
R_{1p} & \leq I(X_1;Y_1|S_{1s}, S_{2s}), \\
R_{1s}+R_{1p}& \leq I(X_1;Y_1|S_{2s}), \\
R_{2p} & \leq I(X_2;Y_2|S_{1s}, S_{2s}) \mbox{ and }\\
R_{2s}+R_{2p}& \leq I(X_2;Y_2|S_{1s}).
\end{align}

For each $n$, similarly to \cite[Theorem 5]{Shannon:1961}, by treating the interference channel $p^n(y_1^n, y_2^n|x_1^n,x_2^n)$, which is a product channel of $p(y_1,y_2|x_1,x_2)$, as a memoryless channel, we conclude that the rates satisfying the following conditions are achievable for any $n$:
\begin{align}
R_{1p} & \leq \frac{1}{n}I(X_1^n;Y_1^n|S_{1s}^n, S_{2s}^n), \\
R_{1s}+R_{1p}& \leq \frac{1}{n}I(X_1^n;Y_1^n|S_{2s}^n), \\
R_{2p} & \leq \frac{1}{n}I(X_2^n;Y_2^n|S_{1s}^n, S_{2s}^n) \mbox{ and }\\
R_{2s}+R_{2p}& \leq \frac{1}{n}I(X_2^n;Y_2^n|S_{1s}^n),
\end{align}
i.e., any rate quadruplet $(R_{1s}, R_{1p}, R_{2s}, R_{2p}) \in \mathcal{G}^n$ is achievable. By the definition of the capacity region, the limiting points of $\mathcal{G}^n$ are also achievable, and thus, we have proved the achievability of all the points in $\mathcal{C}_I$.

We next prove the converse. For any $\left(2^{n R_{1s}}, 2^{n R_{1p}}, 2^{n R_{2s}}, 2^{nR_{2p}}, n\right)$ code, denote its input to the channel as random variables $X_1^n$ and $X_2^n$ and the output of the channel as random variables $Y_1^n$, $Y_2^n$.

Arbitrarily choose $M_{1s} \overset{\triangle}{=} 2^{nR_{1s}}$ $n$-letter sequences $u^{1s}_1, u^{1s}_2, \cdots, u^{1s}_{M_{1s}}$ all in $\mathcal{X}_1^n$, and $M_{2s}\overset{\triangle}{=} 2^{nR_{2s}}$ $n$-letter sequences $u^{2s}_1, u^{2s}_2, \cdots, u^{2s}_{M_{2s}}$ all in $\mathcal{X}_2^n$. Form a one-to-one correspondence between $W_{1s}$, $W_{2s}$ and $S_{1s}^n$, $S_{2s}^n$, respectively by
\begin{align}
p^n(S_{1s}^n=u^n|W_{1s}=w_{1s})&=\left\{
\begin{array}{ll}
1 & \text{ if } u^n=u^{1s}_{w_{1s}}, \quad w_{1s}=1,2,\cdots, M_{1s} \\
0 & \text{ otherwise }
\end{array} \right. \label{test3}\\
p^n(S_{2s}^n=u^n|W_{2s}=w_{2s})&=\left\{
\begin{array}{ll}
1 & \text{ if } u^n=u^{2s}_{w_{2s}}, \quad w_{2s}=1,2,\cdots, M_{2s}\\
0 & \text{ otherwise }
\end{array} \right. \label{test5}
\end{align}

By Fano's inequality \cite{Cover:book}, we have
\begin{align}
n R_{1p}&=H(W_{1p})=H(W_{1p}|W_{1s}, W_{2s})\\
&=I(W_{1p};Y_1^n|W_{1s}, W_{2s})+H(W_{1p}|Y_1^n, W_{1s}, W_{2s})\\
& \leq I(W_{1p};Y_1^n|W_{1s}, W_{2s})+H(W_{1p}|Y_1^n, W_{2s})\\
& \leq I(W_{1p};Y_1^n|W_{1s}, W_{2s})+n \delta(P_e^n) \label{use_delta}\\
& \leq I(X_1^n;Y_1^n|W_{1s}, W_{2s})+n \delta(P_e^n)\label{data_proc2}\\
& = I(X_1^n;Y_1^n|S_{1s}^n, S_{2s}^n)+n \delta(P_e^n) \label{test_channel}
\end{align}
where $\delta(x)$ in (\ref{use_delta}) is a non-negative function approaching zero as $x \rightarrow 0$, (\ref{data_proc2}) follows from data processing inequality \cite{Cover:book} because the distributions factor as $p(w_{1p})p(w_{1s})p(x_1^n|w_{1p},w_{1s})\break p(w_{2p})p(w_{2s})p(x_2^n|w_{2p},w_{2s})p(y_1^n|x_1^n,x_2^n)$ and satisfy the Markov chain relationship $(W_{1p}, W_{1s}) \rightarrow (X_1^n, W_{2s}) \rightarrow Y_1^n$, and  (\ref{test_channel}) follows from the definitions of the sequences $S_{1s}^n$ and $S_{2s}^n$ in (\ref{test3}) and (\ref{test5}), respectively. We also have
\begin{align}
nR_{1s}+nR_{1p}&=H(W_{1s}, W_{1p})=H(W_{1s}, W_{1p}|W_{2s})\\
& =I(W_{1s}, W_{1p};Y_1^n|W_{2s})+H(W_{1s}, W_{1p}|Y_1^n, W_{2s})\\
& \leq I(W_{1s}, W_{1p};Y_1^n|W_{2s})+n \delta(P_e^n)\\
& \leq I(X_1^n;Y_1^n|W_{2s})+n \delta(P_e^n) \label{data_proc1}\\
&= I(X_1^n;Y_1^n|S_{2s}^n)+n \delta(P_e^n) \label{test_channel2}
\end{align}
where (\ref{data_proc1}) follows from the same reason as (\ref{data_proc2}),  and (\ref{test_channel2}) follows from the same reason as (\ref{test_channel}).

Similarly, we have
\begin{align}
nR_{2p} & \leq I(X_2^n;Y_2^n|S_{1s}^n, S_{2s}^n)+n \delta(P_e^n)\\
nR_{2s}+nR_{2p} & \leq I(X_2^n;Y_2^n|S_{1s}^n)+n \delta(P_e^n).
\end{align}
Hence, we have proved that for all $n$,
\begin{align}
(R_{1s}-\delta(P_e^n), R_{1p}-\delta(P_e^n), R_{2s}-\delta(P_e^n), R_{2p}-\delta(P_e^n)) \in \mathcal{G}^n.
\end{align}
Since the region $\mathcal{C}_I$ as defined in (\ref{ci}) contains $\mathcal{G}^n$ for every $n$ \cite[Theorem 5]{Shannon:1961}, we have
\begin{align}
(R_{1s}-\delta(P_e^n), R_{1p}-\delta(P_e^n), R_{2s}-\delta(P_e^n), R_{2p}-\delta(P_e^n)) \in \mathcal{C}_I \label{include}
\end{align}
for all $n$. For codes where $P_e^n \rightarrow 0$ as $n \rightarrow \infty$, we have
\begin{align}
(R_{1s}, R_{1p}, R_{2s}, R_{2p}) \in \mathcal{C}_I \label{closed}
\end{align}
since $\mathcal{C}_I$ is closed \cite[Theorem 5]{Shannon:1961}.
This concludes the converse part of the proof.

\section{Proof of Theorem \ref{sepnew}} \label{proof_theorem_deterministic}
The achievability part of the proof is straightforward. If (\ref{si_int}) holds, then there exists a rate quadruplet $(R_{1s}, R_{1p}, R_{2s}, R_{2p})$ in the interior of $\mathcal{C}$ such that $H(V_k) \leq R_{ks}$ and $H(U_k|V_k) \leq R_{kp}$ for $k=1,2$. Transmitter $k$ first compresses $V_k$ into index $W_{ks}$ with rate $H(V_k)$, and then $U_k|V_k=v_k$ into index $W_{kp}(v_k)$ into rate $H(U_k|V_k)$, for all $v_k$ in the typical set. Then the indices can be transmitted reliably over the channel since $(R_{1s}, R_{1p}, R_{2s}, R_{2p})$ is in the capacity region of the underlying interference channel with message side information $W_{1s}$ at Receiver 2 and $W_{2s}$ at Receiver 1.

To prove the converse, we write
\begin{align}
nH(U_1|V_1)&=H(U_1^n|V_1^n)=H(U_1^n|V_1^n, V_2^n)\\
&=I(U_1^n;Y_1^n|V_1^n, V_2^n)+H(U_1^n|Y_1^n, V_1^n, V_2^n)\\
& \leq I(U_1^n;Y_1^n|V_1^n, V_2^n)+H(U_1^n|Y_1^n, V_2^n)\\
& \leq I(U_1^n;Y_1^n|V_1^n, V_2^n)+n \delta(P_e^n) \label{fano1}\\
& \leq I(X_1^n;Y_1^n|V_1^n, V_2^n)+n \delta(P_e^n) \label{data_proc3}
\end{align}
where (\ref{fano1}) follows from Fano's inequality and $\delta(x)$ is a non-negative function approaching zero as $x \rightarrow 0$, and (\ref{data_proc3}) follows from the data processing inequality, in other words, from the Markov chain relationship $(U_1^n, V_1^n) \rightarrow (X_1^n, V_2^n) \rightarrow Y_1^n$. We can also write
\begin{align}
nH(V_1)+nH(U_1|V_1)&=n H(U_1, V_1)\\
& =n H(U_1) \label{deterministicside}\\
& =H(U_1^n)\\
&=H(U_1^n|V_2^n)\\
&=I(U_1^n;Y_1^n|V_2^n)+H(U_1^n|Y_1^n, V_2^n)\\
& \leq I(U_1^n;Y_1^n|V_2^n)+n \delta(P_e^n)\label{fano2}\\
& \leq I(X_1^n;Y_1^n|V_2^n)+n \delta(P_e^n) \label{data_proc4}
\end{align}
where (\ref{deterministicside}) follows because $V_1$ is a deterministic function of $U_1$, and (\ref{fano2}) follows from Fano's inequality, and (\ref{data_proc4}) follows from the same reasoning as applied to (\ref{data_proc3}). Similarly, we have
\begin{align}
nH(U_2|V_2)& \leq I(X_2^n;Y_2^n|V_1^n, V_2^n)+n \delta(P_e^n) \label{similar1} \mbox{ and } \\
nH(V_2)+n H(U_2|V_2) & \leq I(X_2^n;Y_2^n|V_1^n)+n \delta(P_e^n). \label{similar2}
\end{align}
Hence, from (\ref{data_proc3}), (\ref{data_proc4}), (\ref{similar1}) and (\ref{similar2}), we have
\begin{align}
(H(V_1)-\delta(P_e^n), H(U_1|V_1)-\delta(P_e^n), H(V_2)-\delta(P_e^n), H(U_2|V_2)-\delta(P_e^n)) \in \mathcal{G}^n
\end{align}
which by the same reasoning as applied to (\ref{include}) and (\ref{closed}), for codes where $P_e^n \rightarrow 0$ as $n \rightarrow \infty$, we have
\begin{align}
(H(V_1), H(U_1|V_1), H(V_2), H(U_2|V_2)) \in \mathcal{C}_I
\end{align}
which concludes the proof.

\section{Proof of Lemma \ref{message_region}} \label{proof_message_region}
Due to the fact that the proof of this lemma is very similar to the proof of the capacity region in \cite{Liu_Goldsmith:2008ISIT}, we omit certain details. For notational convenience, denote
the channel of $p(y_1|x_1)$ as $\bar{V}_1$ and the channel $p(y_2|x_1,x_2)$ as $\bar{V}_2$, where
\begin{align}
\bar{V}_1(a|b)&=\text{Pr}\{Y_1=a|X_1=b\},
\end{align}
and
\begin{align}
\bar{V}_2(c|b,d)&=\text{Pr}\{Y_2=c|X_1=b,X_2=d\}.
\end{align}
\subsection{Converse Result}
The converse result derived in this subsection is valid for any Z-interference channel satisfying Condition 1. The tool that we are using comes from the following lemma.
\begin{lemma} \cite[page 314, eqn (3.34)]{Csiszar:book} \label{nletterdifference}

For any $n$, and any random variables $Y^n$ and $Z^n$ and $W$, we have
\begin{align}
H(Z^n|W)-&H(Y^n|W)=\sum_{i=1}^n (H(Z_i|Y^{i-1}, Z_{i+1}, Z_{i+2}, \cdots, Z_n, W)\nonumber\\
&\hspace{0.43in}-H(Y_i|Y^{i-1}, Z_{i+1}, Z_{i+2}, \cdots, Z_n, W) ).
\end{align}
\end{lemma}


Since the rate triplets $(R_{1s}, R_{1p}, R_{2p})$ is achievable, there exist
two sequences of codebooks 1 and 2, denoted by $\mathcal{C}_1^n$ and $\mathcal{C}_2^n$, of rate $R_{1s}+ R_{1p}$ and $R_{2p}$, and probability of error less than $\epsilon_n$, where $P_e^n \rightarrow 0$ as $n \rightarrow \infty$.
Let us define $X_1^n$ and $X_2^n$ be uniformly distributed on codebooks 1 and 2, respectively. Let $Y_1^n$ be connected via $\bar{V}_1^n$ to $X_1^n$, $Y_2^n$ be connected via $\bar{V}_2^n$ to $X_1^n$ and $X_2^n$.

We start the converse with Fano's inequality \cite{Cover:book},
\begin{align}
nR_{1p}&=H(W_{1p})\\
&\leq I(W_{1p};Y_1^n)+n \delta(P_e^n)\\
&\leq I(W_{1p};Y_1^n|W_{1s})+n \delta(P_e^n) \label{independent}\\
&= H(Y_1^n|W_{1s})-H(Y_1^n|W_{1s}, W_{1p}, X_1^n)+n \delta(P_e^n) \label{deterministic}\\
& = H(Y_1^n|W_{1s})-H(Y_1^n|X_1^n)+n \delta(P_e^n) \label{markov}\\
&=H(Y_1^n|W_{1s})-\sum_{i=1}^n H(Y_{1i}|X_{1i})+ n \delta(P_e^n) \label{memoryless}
\end{align}
where (\ref{independent}) follows from the fact that $W_{1s}$ and $W_{1p}$ are independent, (\ref{deterministic}) follows from the fact that without loss of generality, we may consider deterministic encoders, (\ref{markov}) follows from the Markov chain relationship $(W_{1s}, W_{1p}) \rightarrow X_1^n \rightarrow Y_1^n$, and (\ref{memoryless}) follows from the memoryless nature of $\bar{V}_1^n$. We also have
\begin{align}
nR_{1s}+n R_{1p}&=H(W_{1p}, W_{1s})\\
&\leq I(W_{1p}, W_{1s};Y_1^n)+n \delta(P_e^n)\\
&\leq I(X_1^n;Y_1^n)+n \delta(P_e^n)\label{dataprocessing}\\
&\leq \sum_{i=1}^n I(X_{1i};Y_{1i})+ n \delta(P_e^n) \label{sumlater}
\end{align}
where (\ref{dataprocessing}) follows from the data processing inequality \cite{Cover:book}. Furthermore, we have
\begin{align}
n R_{2p} =H(W_{2p})&=H(W_{2p}|W_{1s}) \label{independent2}\\
& \leq I(W_{2p};Y_2^n|W_{1s})+n \delta(P_e^n)\\
& \leq I(X_2^n;Y_2^n|W_{1s})+n \delta(P_e^n) \label{dataprocessing2}\\
& = H(Y_2^n|W_{1s})-H(Y_2^n|X_2^n, W_{1s})+n\delta(P_e^n)\\
& \leq \sum_{i=1}^n H(Y_{2i})-H(Y_2^n|X_2^n, W_{1s})+n\delta(P_e^n) \label{conditionreduce}\\
& \leq n \tau -H(Y_2^n|X_2^n, W_{1s})+n\delta(P_e^n) \label{usemax}
\end{align}
where (\ref{independent2}) follows from the independence of $W_{2p}$ and $W_{1s}$, (\ref{dataprocessing2}) follows from the Markov chain relationship $W_{2p} \rightarrow (X_1^n, W_{1s}) \rightarrow Y_2^n$, (\ref{conditionreduce}) follows from the fact that conditioning reduces entropy, and (\ref{usemax}) follows from the definition of $\tau$ in (\ref{definetau}).

Let us define another channel, $\hat{V}_2: \mathcal{X}_1 \rightarrow \mathcal{Y}_2$, as
\begin{align}
\hat{V}_2(t|x_1)=V_2(t|x_1, \bar{x}_2),
\end{align}
where $\bar{x}_2$ is an arbitrary element in $\mathcal{X}_2$. Further, let us define another sequence of random variables, $T^n$, which is connected via $\hat{V}_2^n$, the memoryless channel $\hat{V}_2$ used $n$ times, to $X_1^n$, i.e., $T_i \rightarrow X_{1i} \rightarrow T_{\{i\}^c}, X_{1\{i\}^c}, X_2^n, Y_1^n, Y_2^n$. Also define $\bar{x}_2^n$ as the $n$-sequence with $\bar{x}_2$ repeated $n$ times. It is easy to see that
\begin{align}
H(Y_2^n|X_2^n, W_{1s})&=\sum_{x_2^n \in \mathcal{C}_2^n} \sum_{w=1}^{2^{nR_{1s}}} \frac{1}{2^{nR_{1s}}} \frac{1}{2^{nR_{2p}}} H(Y_2^n|X_2^n=x_2^n, W_{1s}=w)\\
&=\sum_{w=1}^{2^{nR_{1s}}} \frac{1}{2^{nR_{1s}}} H(Y_2^n|X_2^n=\bar{x}_2^n, W_{1s}=w) \label{usecon1}\\
&=\sum_{w=1}^{2^{nR_{1s}}} \frac{1}{2^{nR_{1s}}} H(T^n| W_{1s}=w) \label{useT}\\
&= H(T^n|W_{1s}) \label{sumlater2}
\end{align}
where (\ref{usecon1}) follows from the fact that the channel under consideration satisfies condition \ref{shiftn}, and (\ref{useT}) follows from the definition of $T^n$.

By applying Lemma \ref{nletterdifference}, we have
\begin{align}
H(T^n|W_{1s})-&H(Y_1^n|W_{1s})=\sum_{i=1}^n H(T_i|Y_1^{i-1}, T_{i+1}, T_{i+2}, \cdots, T_n, W_{1s}) \nonumber\\
&\hspace{0.43in}-H(Y_{1i}|Y_1^{i-1}, T_{i+1}, T_{i+2}, \cdots, T_n, W_{1s}). \label{km1}
\end{align}
Furthermore, since conditioning reduces entropy, we can write
\begin{align}
H(Y_1^n|W_{1s})&=\sum_{i=1}^n H(Y_{1i}|Y_1^{i-1}, W_{1s})\geq \sum_{i=1}^n H(Y_{1i}|Y_{1}^{i-1}, T_{i+1}, T_{i+2}, \cdots, T_n, W_{1s}). \label{km2}
\end{align}
Define the following auxiliary random variables,
\begin{align}
W_i=Y_{1}^{i-1}, T_{i+1}, T_{i+2}, \cdots, T_n, W_{1s}, \qquad i=1,2,\cdots,n.
\end{align}
Further define $Q$ as a random variable that is uniform on the set $\{1,2, \cdots, n\}$ and independent of everything else. Also, define the following auxiliary random variables:
\begin{align}
W=(W_Q,Q), \quad X_1=X_{1Q}, \quad Y_1=Y_{1Q} \mbox{ and }\quad T=T_Q.
\end{align}
Then, from (\ref{km1}) and (\ref{km2}), we have
\begin{align}
n^{-1} \left(H(T^n|W_{1s})-H(Y_1^n|W_{1s}) \right)&=H(T|W)-H(Y_1|W) \mbox{ and } \label{korner1}\\
 n^{-1}H(Y_1^n|W_s) &\geq H(Y_1|W). \label{korner2}
\end{align}
Due to the memoryless nature of $\bar{V}_1^n$ and $\hat{V}_2^n$, the fact that $Q$ is independent of everything else, and the Markov chain relationship $T_i \rightarrow X_{1i} \rightarrow Y_{1i}$, for $i=1,2,\cdots,n$, the joint distribution of $W$, $X_1$, $Y_1$, $T$ satisfies
\begin{align}
p(w,x_1, y_1, t)&=p(w)p(x_1|w)V_1(y_1|x_1)V_2(t|x_1, \bar{x}_2). \label{distribution}
\end{align}
From (\ref{korner1}) and (\ref{korner2}), we may conclude that there exists a number $\gamma \geq 0$ such that
\begin{align}
\frac{1}{n}H(T^n|W_{1s})=H(T|W)+\gamma, \quad \frac{1}{n}H(Y_1^n|W_{1s})=H(Y_1|W)+\gamma. \label{ttconverse}
\end{align}
By combining (\ref{memoryless}), (\ref{sumlater}), (\ref{usemax}), (\ref{sumlater2}),  (\ref{distribution}), and (\ref{ttconverse}), and allowing $n \rightarrow \infty$, we obtain the following converse result:
for any Z-interference channel that satisfies Condition 1 and the case where Receiver 2 has side information $W_{1s}$, the achievable rate triplets $(R_{1s}, R_{1p}, R_{2p})$ must satisfy
\begin{align}
R_{1p} & \leq H(Y_1|W)+\gamma-H(Y_1|X_1), \label{upper1}\\
R_{1s}+R_{1p} & \leq I(X_1;Y_1) \mbox{ and } \\
R_{2p} &\leq \tau-H(T|W)-\gamma, \label{upper2}
\end{align}
for some number $\gamma \geq 0$ and distribution $p(w)p(x_1|w)$, where the mutual informations and entropies are evaluated using
$p(w,x_1, y_1,t)=p(w)p(x_1|w)V_1(y_1|x_1)V_2(t|x_1, \bar{x}_2)$.

\subsection{Achievability Result}
The achievability result derived in this subsection is valid for any Z-interference channel.
We design a codebook at Transmitter 1 such that the inner codebook carries the side information at the Receiver 2, i.e., $W_{1s}$, and part of $W_{1p}$, and the outer codebook carries the remaining part of $W_{1p}$. More specifically, the inner codebook is of rate $R_{1s}+\gamma$, and the outer codebook is of rate $R_{1p}-\gamma$. Then, we have the achievable rate region as the union over all $p(w)p(x_1|w)p(x_2)$ of
\begin{align}
R_{1p} & \leq H(Y_1|W)+\gamma-H(Y_1|X_1) \label{lowerapp1}\\
R_{1s}+R_{1p} & \leq I(X_1;Y_1)\\
R_{2p} & \leq I(X_2;Y_2|W) \mbox{ and }\\
R_{2p} & \leq I(W, X_2;Y_2)-\gamma, \label{lowerapp2}
\end{align}
 where the mutual informations are evaluated using $\break p(w,x_1,x_2,y_1,y_2)=p(w)p(x_1|w)p(x_2)V_1(y_1|x_1) V_2(y_2|x_1,x_2)$.

\subsection{Capacity Region}
Making use of Conditions 1 and 2 in the exact same way as in \cite[Section V]{Liu_Goldsmith:2008ISIT}, we can show that the converse result in (\ref{upper1})-(\ref{upper2}) and the achievability result in (\ref{lowerapp1})-(\ref{lowerapp2}) are the same for Z-interference channels satisfying Conditions 1 and 2, and hence the capacity region $\mathcal{C}_I$ in this case is given in Lemma \ref{message_region}.

\bibliographystyle{unsrt}
\bibliography{ref}

\end{document}